\documentclass[
  a4paper,
  UKenglish,
  cleveref,
  autoref,
  thm-restate
]{lipics-v2021}

\pdfoutput=1
\hideLIPIcs
\nolinenumbers
\usepackage[table]{xcolor}
\usepackage{graphicx}

\title{Towards Decentralized Searcher Competition in MEV Markets}
\titlerunning{Towards Decentralized Searcher Competition}

\author
  {Roozbeh Sarenche}
  {COSIC, KU Leuven, Leuven, Belgium}
  {roozbeh.sarenche@esat.kuleuven.be}
  {}
  {Supported by the Flemish Government through the Cybersecurity Research Program Flanders with grant number VOEWICS02.}

\author
  {Yunwen Liu}
  {COSIC, KU Leuven, Leuven, Belgium}
  {yunwen.liu@esat.kuleuven.be}
  {}
  {Partially supported by KU Leuven STG/25/039 and CTVG R\&D grant (QJW-410118-LIU-CTVG).}

\authorrunning{R. Sarenche and Y. Liu}

\Copyright{Roozbeh Sarenche and Yunwen Liu}

\ccsdesc{Theory of computation~Algorithmic game theory and mechanism design}
\ccsdesc{Security and privacy~Economics of security and privacy}

\keywords{
Maximal extractable value,
MEV,
searcher competition,
mechanism design,
Sybil resistance,
Shapley value,
reward decentralization
}

\begin{document}

\maketitle

\begin{abstract}
Centralization in maximal extractable value (MEV) markets is a significant concern for blockchain systems, as persistent concentration of economic power can weaken competition, reduce openness, and undermine the decentralization goals of permissionless protocols. While much of the existing analysis has focused on builders, validators, and block-building markets, this paper studies centralization from the perspective of searcher competition. We develop a heterogeneous model in which searchers differ in opportunity coverage and execution efficiency, and we analyze how auction design affects fairness, decentralization, and security among searchers competing for the same MEV opportunity.

To evaluate searcher competition, we introduce two metrics: a Shapley-weighted Jain fairness index, which measures whether rewards are proportional to searchers' marginal contributions, and an expected-reward Herfindahl--Hirschman Index (HHI), which measures concentration in long-run searcher rewards. Using these metrics, we first analyze the standard first-price, winner-take-all auction as a benchmark. Our analysis shows that, under searcher heterogeneity, first-price competition can reward rank dominance rather than marginal contribution, leading to concentrated rewards and weaker contribution-adjusted fairness.

Motivated by these limitations, we propose an entry-filtered Shapley-capped auction mechanism that distributes searcher rewards more fairly and broadly among admitted high-quality submissions. Designing such a mechanism in a permissionless blockchain environment is challenging: searchers may create Sybil identities by submitting copied or degraded versions of the same execution strategy, and validators may collude with searchers to increase joint payoff. We address these concerns through Bayesian security constraints for copied-code Sybil deviations and validator--searcher coalition deviations, and show that suitable parameter choices make the  Shapley-capped mechanism secure against both.

We complement the theoretical analysis with numerical benchmarks and Ethereum on-chain data. The theoretical benchmarks show that the proposed mechanism is especially effective in centralized opportunity classes, where first-price competition concentrates rewards among a small set of searchers.
They further indicate that the proposed mechanism is a potential improvement over first-price competition in empirically motivated opportunity types.
\end{abstract}

\section{Introduction}

In Ethereum's Maximal Extractable Value (MEV) supply chain, \textit{searchers} are specialized actors who identify profitable execution opportunities and submit transactions or bundles to capture them~\cite{DBLP:conf/sp/DaianGKLZBBJ20,flashbotsSearchingPostMerge}. 
Their submissions enter auction-like competition for block inclusion, where searchers bid part of the generated value to block producers---builders and validators.

% Prior work suggests that searcher-level MEV extraction is not a symmetric market, but is shaped by specialized bot behavior, unequal access to opportunities, and heterogeneous competitive pressure
% Prior work suggests that searcher-level MEV extraction is far from a symmetric market: searchers differ in opportunity access, latency, information, capital, and execution quality, and these differences are amplified by specialized bot behavior and heterogeneous competitive pressure~\cite{DBLP:conf/imc/WeintraubTNS22,DBLP:conf/ndss/LuoLLHLMSC26}. 
% For atomic MEV strategies, including sandwiching and arbitrage between decentralized exchanges, early studies already observed sophisticated bot specialization and intense competition~\cite{DBLP:conf/sp/DaianGKLZBBJ20}. 
% More recently, Heimbach et al.~\cite{DBLP:conf/sp/HeimbachPS24} showed that searcher activity can become highly concentrated in arbitrage markets: among searchers exploiting price discrepancies between centralized and decentralized exchanges, eleven searchers account for more than 80\% of the identified arbitrage volume, with the two largest searchers alone accounting for 49.2\%. 
% These findings suggest that profitable searcher markets can become highly skewed when opportunities require specialized infrastructure, private information, capital, or privileged access. 

Prior work suggests that searcher-level MEV extraction is far from a symmetric market. 
Searchers differ in opportunity access, latency, information, capital, and execution quality, while specialized bot behavior and heterogeneous competitive pressure further amplify these differences~\cite{DBLP:conf/ndss/LuoLLHLMSC26,DBLP:conf/imc/WeintraubTNS22}. 
For atomic MEV strategies, including sandwiching and decentralized-exchange arbitrage, early studies already documented sophisticated bot specialization and intense competition~\cite{DBLP:conf/sp/DaianGKLZBBJ20}. 
More recent evidence shows that concentration can be severe: in centralized--decentralized exchange arbitrage, Heimbach et al.~\cite{DBLP:conf/sp/HeimbachPS24} find that eleven searchers account for more than 80\% of identified arbitrage volume, with the two largest searchers alone accounting for 49.2\%. 

These observations suggest that searcher markets can become highly skewed in the presence of players with specialized infrastructure, private information, capital, or privileged access. 
This skewness, together with current searcher-auction design, has consequences beyond searcher profits and can become a bottleneck for blockchain decentralization. 
Persistent reward concentration weakens open competition and raises effective barriers to entry.
This concern is amplified on high-throughput chains where MEV competition often takes the form of large-scale speculative transaction submission~\cite{wang2026blockspace}. 
It can also reinforce builder-side concentration when valuable searcher flow is routed privately or integrated with particular builders~\cite{DBLP:conf/sp/HeimbachPS24,DBLP:conf/aft/OzSTM24,DBLP:conf/fc/PaiR24}. 
Searcher concentration can then propagate through the whole MEV supply chain: dominant searchers become preferred counterparties, builders with access to their flow gain an advantage, and validators face a block-building market shaped by privileged order-flow channels.

Most work on MEV concentration has focused on the block-building layer, analyzing how private order flow, vertical integration, latency advantages, and strategic bidding drive builder concentration under proposer-builder separation~\cite{DBLP:conf/fc/BahraniGR24a,DBLP:conf/aft/GuptaPR23,DBLP:conf/sp/HeimbachPS24,DBLP:conf/aft/OzSTM24,pahari2025exclusive,DBLP:conf/fc/PaiR24,DBLP:conf/www/WangHZHW025,DBLP:conf/icbc2/WuTLV24,DBLP:conf/icaif/WuTLV24,DBLP:conf/ecai/WuTLV25,DBLP:conf/sp/YangN025,zhang2026order}. 
Within this view, searchers are typically modeled merely as sources of transaction bundles or private flow that confer competitive advantages to builders. 
Consequently, the internal mechanics of the searcher auction market have received considerably less attention. 
While recent studies have begun to examine searcher mechanism design and bidding strategies~\cite{DBLP:conf/aft/MamageishviliSS24}, the specific implications of these auction formats for fairness and decentralization in the searcher market remain largely unexplored.
This motivates our central question: how fair and decentralized is the current first-price, winner-take-all searcher auction, and can an alternative mechanism design improve reward decentralization while preserving the rigorous security guarantees required for permissionless MEV competition?

%a game-theoretic framework that treats searcher competition as a measurable market structure, linking empirical patterns of MEV extraction to mechanism-level objectives such as open participation, decentralization, and mitigation of persistent searcher advantages.

Decentralizing rewards in a permissionless environment is fundamentally difficult.
If a mechanism pays multiple submissions, a searcher may split a single execution strategy into multiple identities, submitting copied or intentionally degraded versions of the same code.
Consequently, any mechanism that naively rewards participation inherently invites Sybil attacks.
Moreover, validators and searchers have strong incentives to coordinate outside the protocol if identity splitting can increase their joint payoff.
Therefore, a decentralized searcher auction must solve two problems simultaneously: it must distribute rewards more broadly and fairly, while strictly preserving the security guarantees required for permissionless MEV execution.

This paper studies that tradeoff. 
We develop a model of heterogeneous searcher competition, analyze the standard first-price auction, and introduce a Shapley-based mechanism as a more decentralized alternative. 
The main message is that the current winner-take-all design is secure against copied-code deviations, but can be economically centralizing under realistic searcher heterogeneity. 
By contrast, a carefully capped and burned Shapley-based mechanism can improve contribution-adjusted fairness and reward decentralization, while controlling Sybil and validator--searcher coalition incentives through explicit Bayesian security constraints.

\subsection{Our Contribution}
Our contribution includes:

\smallskip\noindent\textbf{System model and evaluation criteria.}
We introduce a model of searcher competition in which searchers differ in both opportunity coverage and execution quality. 
For each opportunity class, a searcher may or may not discover and submit for the opportunity, and conditional on participation its submitted code has a certified normalized net value. 
This lets us separate two important sources of searcher advantage: seeing more opportunities and executing a given opportunity more efficiently. 
We then define economic and security criteria for evaluating searcher auctions. 
The economic criteria are a Shapley-weighted Jain fairness index, which measures whether rewards are aligned with marginal contribution, and an expected-reward HHI, which measures reward concentration. 
The security criteria are Bayesian copied-code Sybil resistance and Bayesian validator--searcher coalition resistance.

\smallskip\noindent\textbf{Analysis of the first-price auction.}
We analyze the first-price auction as the benchmark for current winner-take-all searcher competition. 
Under monotone searcher-specific bidding, copied or degraded versions of a searcher's own code cannot improve either the searcher's payoff or the joint payoff of the searcher and validator. 
Thus, the first-price auction has strong copied-code security properties. 
Economically, however, its rewards are governed by rank dominance: searchers are rewarded for being the highest-bidding active participant, not for their marginal contribution. 
We show that this distinction becomes important under heterogeneity. 
When execution quality is unequal, first-price rewards can concentrate on a small dominant class; in the limiting case of a single dominant searcher, the effective number of rewarded searchers approaches one.

\smallskip\noindent\textbf{Shapley-Capped Auction.}
We propose the Shapley-Capped Auction as an alternative design. 
The mechanism still executes the highest certified-value submission, but when the number of admitted submissions is below a cap, it distributes retained rewards according to the Shapley value of the validator-dependent max game. 
When the number of admitted submissions exceeds the cap, it switches to a conservative fallback branch and burns the residual. 
The cap and burn parameters are used to control the incentives created by reward sharing. 
We prove budget feasibility and provide a Bayesian security characterization showing that it is sufficient to consider a bounded range of Sybil sizes while optimizing over all feasible copied or degraded submissions. This yields a feasible region for the security parameters under which the mechanism satisfies copied-code Sybil resistance and validator--searcher coalition resistance.

\smallskip\noindent\textbf{Real-world data and empirical calibration.}
We analyze Ethereum arbitrage data to connect the model to observed searcher behavior. 
The goal is to estimate practical patterns of searcher participation, profit concentration, and opportunity coverage, and then use these patterns to evaluate the mechanism under realistic parameters. 
This empirical component is not meant to fully identify every real-world searcher entity, since a single entity may control multiple contracts, but it provides a data-driven basis for calibrating the participation and execution-quality primitives used in the theoretical and numerical analysis.

\subsection{Related Work}
\smallskip\noindent\textbf{Searcher's MEV activities.}
To produce competitive bundles, searchers rely on a combination of advantageous resources: efficient algorithms and heuristics to detect MEV opportunities, substantial capital to execute transactions across domains, and information asymmetries to shield lucrative opportunities from competitors. 
The strategic landscape has grown considerably more complex over recent years. 
Early MEV extraction was focused on  \textit{atomic} strategies, where all transactions constituting an attack execute within a single chain or block, guaranteeing all-or-nothing execution. 
Prominent examples include cyclic arbitrage across decentralized exchanges (DEX) and sandwich attacks on pending user trades~\cite{DBLP:conf/sp/DaianGKLZBBJ20,DBLP:conf/sp/QinZG22}.
Later, the proliferation of Layer-2 systems and the frictions between heterogeneous market domains have increasingly pushed searchers toward \textit{non-atomic} strategies, giving rise to arbitrage opportunities spanning multiple layers, chains, and off-chain venues such as centralized exchanges (CEX)~\cite{gogol2024cross,DBLP:journals/pomacs/OzTSMGRM25,DBLP:conf/ccs/TorresMWNS24}. 

\smallskip\noindent\textbf{Centralization in the MEV market.}
Block building on Ethereum has become highly concentrated: the leading builder, Titan, alone accounts for around 50\% of blocks in May 2026. 
A prominent study explaining builder centralization is that of {\"{O}}z et al.~\cite{DBLP:conf/aft/OzSTM24}, which links builders' market share to order-flow diversity and their profitability to access to exclusive order-flow providers.
Recent measurement studies on block builders also suggest that builder dominance is shaped not only by builder-side strategies, but more importantly by access to valuable searcher-generated bundles~\cite{DBLP:conf/sp/HeimbachPS24,DBLP:conf/aft/WuSTP25,DBLP:conf/sp/YangN025}. 
These bundles may derive their value from private information~\cite{pahari2025exclusive}, more effective search and bidding strategies~\cite{DBLP:conf/ndss/LuoLLHLMSC26}, or vertical integration between searchers and builders~\cite{DBLP:conf/fc/PaiR24}.

Compared with the extensive literature on block builders, systematic studies of searchers remain relatively recent.
Searchers are heterogeneous in both strategy and market access. 
At the strategy level, searchers specialize in different forms of MEV extraction, including arbitrage, sandwich, liquidation, and mixed strategies~\cite{DBLP:conf/ndss/LuoLLHLMSC26}. 
At the access level, searchers may rely on public mempool transactions, private or exclusive order flow, MEV-Share signals\footnote{https://github.com/flashbots/mev-share}, and off-chain information. 

Moreover, searcher competition is unevenly distributed across opportunities and actors. 
Mamageishvili et al.~\cite{DBLP:conf/aft/MamageishviliSS24} showed that validator revenue increases with searcher competition: when many searchers identify similar opportunities, competition pushes most of the value to the validator, while searchers retain more value when competition is weak. 
Empirical studies further confirmed that the searcher market is highly skewed toward large or specialized actors. 
For non-atomic arbitrage, Heimbach et al.~\cite{DBLP:conf/sp/HeimbachPS24} and Wu et al.~\cite{DBLP:conf/aft/WuSTP25} found that a small number of searchers account for most of the identified arbitrage volume and value. 

These results suggest that searchers do not compete in a single homogeneous market---some MEV opportunities are highly contested, whereas others are repeatedly captured by searchers with persistent advantages. 
From a game-theoretic perspective, Mamageishvili et al.~\cite{DBLP:conf/aft/MamageishviliSS24} provide an important foundation by modeling searcher competition as a cooperative \textit{submodular} game, where the value function maps each coalition of searcher bundles to the highest block value achievable from that coalition. 
What remains less understood is how to model searcher competition in a way that captures heterogeneous opportunity access, variation in execution quality, and measurable competition intensity across the searcher market.

\smallskip\noindent\textbf{Game-Theoretic Approaches to MEV Auction Design.}
Recent game-theoretic analyses of MEV auctions increasingly emphasize the tension between allocative efficiency, decentralization, and resilience against strategic manipulation~\cite{roughgarden2021transaction,Garimidi2026BeyondWTA}. While traditional winner-take-all procurement auctions inherently drive market concentration, alternative mechanisms that penalize monopoly allocations must strictly account for Sybil attacks in permissionless environments~\cite{Garimidi2026BeyondWTA}. Pan et al.~\cite{Pan2026SybilProof} demonstrate that the second-price auction with symmetric tie-breaking is the only non-wasteful, incentive-compatible, and Sybil-proof direct mechanism in single-parameter settings, severely limiting the design space for secure alternative distributions. To circumvent these constraints and achieve fair surplus allocation without compromising security, recent work applies cooperative game theory directly to order flow auctions (OFAs)~\cite{DBLP:conf/aft/MamageishviliSS24}. Specifically, Shapley value-based frameworks have been proposed to equitably redistribute matchmaking revenues among transaction creators according to their marginal contributions, utilizing randomized approximation to handle the computational complexity of the MEV supply chain~\cite{rasheed2025shapley}. Complementing these static allocation approaches, Braga et al.~\cite{Braga2024DynamicMEV} propose a dynamic MEV-sharing mechanism that adjusts the protocol-level extraction rate over time to balance participation incentives between users and MEV extractors. Together, this literature highlights a critical shift from purely competitive block building toward robust, Sybil-proof mechanisms and cooperative revenue-sharing models.

\section{Model, Metrics, and Mechanisms}
\label{sec:model}

This section introduces a heterogeneous model of MEV searcher competition. Searchers differ in both opportunity coverage and execution efficiency. We evaluate mechanisms along two economic dimensions---contribution-adjusted fairness and reward decentralization---and two security dimensions---copied-code Sybil resistance and validator--searcher coalition resistance.

\subsection{System Model}
\label{subsec:system-model}

We model searcher competition for MEV opportunities in a setting with heterogeneous opportunity discovery, heterogeneous execution quality, and permissionless submissions.
We focus on on-chain competition and abstract away from exogenous off-chain or cross-domain factors, such as centralized-exchange price movements and MEV opportunities originating from off-chain systems.

\noindent\textbf{Entities and opportunity classes.}
There is one validator (proposer) and a finite set of potential searchers
\(
\mathcal S=\{1,\ldots,n\}.
\)
MEV opportunities are indexed by \(o\in\mathcal O\). Each opportunity belongs to a class
\(
\tau(o)\in\mathcal C\footnote{A class may represent a broad MEV type, such as arbitrage, liquidation, sandwiching, or a finer subclass, such as arbitrage over a particular family of AMMs or routes.}.\)

\noindent\textbf{Opportunity coverage.}
Searcher \(i\in\mathcal S\) can participate in opportunity \(o\) only if it discovers the opportunity and can submit \emph{execution code} before the deadline. We call such searchers \emph{active searchers} for opportunity \(o\). Let
\(
D_{i,o}\sim\mathrm{Bernoulli}(\alpha_{i,\tau(o)})
\)
denote this event. The parameter \(\alpha_{i,\tau(o)}\in[0,1]\) is searcher \(i\)'s coverage probability for class \(\tau(o)\). The realized active set for opportunity \(o\) is
\[
\mathcal S_{\mathrm{act}}(o)=\{i\in\mathcal S:D_{i,o}=1\},
\qquad
N_o=|\mathcal S_{\mathrm{act}}(o)|.
\]
How opportunities are clustered into classes, and how \(\mathcal S_{\mathrm{act}}(o)\) is empirically identified, is outside the formal model.\footnote{One practical approach is to place submissions in the same active set \(\mathcal S_{\mathrm{act}}\) when executing the highest-net-value submission would reduce the residual profitability of the remaining submissions below a specified threshold.}

\noindent\textbf{Execution code, certified values, and normalization.}
Any execution code is identified by two certified quantities under the canonical simulation rule: its gross realized value and its execution cost. Let \(\tilde v_{i,o}\) denote the gross certified value generated by searcher \(i\)'s code, and let \(\tilde c_{i,o}\) denote its certified execution cost. The cost \(\tilde c_{i,o}\) captures execution losses that are objectively computable from the submitted code, such as base-fee gas costs, realized slippage, and flash-loan fees; it does not include transfers to the validator. We define the certified net value of this code as \(\tilde z_{i,o}=\tilde v_{i,o}-\tilde c_{i,o}\).

For each opportunity \(o\), let \(V_o>0\) denote a public opportunity-scale parameter, known before the submission deadline. This scale represents the commonly observable or commonly estimated size of the opportunity and is used to normalize certified values across opportunities. We normalize these quantities by \(V_o\):
\(
v_{i,o}=\frac{\tilde v_{i,o}}{V_o},
\quad
c_{i,o}=\frac{\tilde c_{i,o}}{V_o},
\quad
z_{i,o}=\frac{\tilde z_{i,o}}{V_o}.
\)
We assume \(V_o\) is chosen so that \(v_{i,o}\le 1\) and \(z_{i,o}\le 1\).

\noindent\textbf{Submissions and copied-code deviations.}
For opportunity \(o\), only active searchers with \(D_{i,o}=1\) can submit. A submission \(b\) consists of execution code and, depending on the auction format, may also include a mechanism-specific message \(m_b\), such as a transfer offer. Each submission has a certified net value \(\tilde z_b\) and certified normalized net value \(z_b=\tilde z_b/V_o\), computed under the canonical simulation rule. In a permissionless setting, a searcher may attempt to increase its payoff by representing one execution code through multiple submissions. We model this through \emph{copied-code Sybil deviations} and denote the finite set of submissions made by active searcher \(i\) by \(\mathcal B_i(o)\).

\begin{definition}[Copied-code Sybil deviation]
Fix an opportunity \(o\), an active searcher \(i\in\mathcal S_{\mathrm{act}}(o)\), and a submission \(b\in\mathcal B_i(o)\) with certified normalized net value \(z_b\). A \(k\)-submission copied-code Sybil deviation from \(b\), with \(k\ge2\), is a collection of submissions
\(
\sigma_{i,o}=(\hat b_1,\ldots,\hat b_k)
\)
by the same real searcher \(i\), with certified normalized net values
\(
(z_{\hat b_1},\ldots,z_{\hat b_k}),
\)
such that
\(
z_{\hat b_\ell}\le z_b\)
for every \(\ell=1,\ldots,k.\) 
The reference submission \(b\) is called the original submission.
\end{definition}

\noindent We denote by \(\mathcal B(o)\) the set of all submissions for opportunity \(o\), including both original submissions and Sybil submissions:
\(
\mathcal B(o)=\bigcup_{i\in\mathcal S_{\mathrm{act}}(o)}\mathcal B_i(o),\)
and \(
t_o=|\mathcal B(o)|.
\)
The validator can include at most one submitted execution code for opportunity \(o\).

\noindent\textbf{Bayesian information structure.}
Before the submission deadline, an active searcher \(i\) observes its own original execution code and its certified normalized net value. It does not observe the realized participation decisions, execution codes, certified values, or mechanism-specific messages of other searchers\footnote{This information structure can be implemented, for example, by an encrypted mempool or commit--reveal protocol in which validators and searchers do not observe other submissions before the epoch deadline for decryption or reveal.}. It only knows the corresponding distributional primitives. For each \(j\neq i\),
\(
D_{j,o}\sim \mathrm{Bernoulli}(\alpha_{j,\tau(o)}).
\)
Conditional on \(D_{j,o}=1\), the certified normalized net value of searcher \(j\)'s original code for opportunity \(o\), denoted by \(z^\star_{j,o}\), is distributed according to the cumulative distribution function \(F_{j,\tau(o)}\):
\(
z^\star_{j,o}\sim F_{j,\tau(o)},
\quad
z^\star_{j,o}\in
[\underline{z}_{j,\tau(o)},\overline{z}_{j,\tau(o)}]\subset(0,1).
\)
The distribution \(F_{j,\tau(o)}\) is searcher- and class-specific: it describes searcher \(j\)'s normalized performance across opportunities of class \(\tau(o)\).
After the submission deadline, the set of submissions is fixed. The mechanism then evaluates the submitted execution codes under the canonical simulation rule and observes their realized certified normalized net values. This interim information structure is the basis for the Bayesian-form security properties discussed later in the paper.

\noindent\textbf{Auction mechanism.}
For each opportunity \(o\), a searcher auction mechanism \(\mathcal M\) takes the submitted execution codes, their certified net values, and any mechanism-specific messages as input. It selects one submission for execution and assigns payoffs to the real searchers and the validator.
For each submission \(b\in\mathcal B(o)\), let \(\tilde z_b=\tilde v_b-\tilde c_b\) denote its certified net value. Let \(m_b\) denote its mechanism-specific message. Write \(\tilde z_o=(\tilde z_b)_{b\in\mathcal B(o)}\) and \(m_o=(m_b)_{b\in\mathcal B(o)}\). The mechanism specifies a winner-selection rule \(\operatorname{Win}_o^{\mathcal M}(\tilde z_o,m_o)\in\mathcal B(o)\), which selects a single winning submission. The validator includes the execution code associated with the selected submission.
The mechanism also specifies payoff rules \(U_i^{\mathcal M}(\tilde z_o,m_o)\) for each \(i\in\mathcal S_{\mathrm{act}}(o)\), and \(U_{\mathrm{val}}^{\mathcal M}(\tilde z_o,m_o)\) for the validator. The payoff \(U_i^{\mathcal M}\) is the total payoff assigned to real searcher \(i\), aggregating over all submissions in \(\mathcal B_i(o)\).

\subsection{Evaluation Criteria}
\label{subsec:metrics}

We evaluate a searcher auction mechanism \(\mathcal M\) along two dimensions: economic performance and security. The economic metrics measure whether rewards are fairly and broadly distributed among real searchers. The security properties require that these rewards cannot be profitably manipulated through copied-code Sybil submissions or validator--searcher coalitions. Let \(U_i^{\mathcal M}\) denote searcher \(i\)'s realized payoff and \(U_{\mathrm{val}}^{\mathcal M}\) the validator's realized payoff. Define searcher \(i\)'s expected reward under \(\mathcal M\) as
\(
R_i^{\mathcal M}=\mathbb E[U_i^{\mathcal M}].
\)

\subsubsection{Economic metrics} \label{subsec:economic-metrics}

The first economic metric measures whether expected rewards are aligned with searchers' contributions. We use the Shapley value as the contribution benchmark. The Shapley value is a standard solution concept from cooperative game theory: it assigns to each participant its average marginal contribution over all possible orders in which participants could join a coalition. In our setting, this provides a natural way to measure how much each searcher contributes to the expected surplus generated by competition for an opportunity.

Let \(\mathcal N=\mathcal S\cup\{\mathrm{val}\}\). For \(A\subseteq\mathcal N\), define the expected value of coalition \(A\) in the validator--searcher game as
\[
W^{\mathrm{VS}}(A)=
\begin{cases}
\mathbb E\!\left[V_o\max_{i\in A\cap\mathcal S:D_{i,o}=1} z^\star_{i,o}\right],
& \mathrm{val}\in A,\\
0,& \mathrm{val}\notin A,
\end{cases}
\]
with the convention that the maximum over an empty set is zero. Searcher \(i\)'s validator--searcher Shapley contribution~\cite{Shapley1953Value} is
\[
\phi_i^{\mathrm{VS}}
=
\sum_{A\subseteq\mathcal N\setminus\{i\}}
\frac{|A|!(n-|A|)!}{(n+1)!}
\bigl(W^{\mathrm{VS}}(A\cup\{i\})-W^{\mathrm{VS}}(A)\bigr).
\]
Assume \(\phi_i^{\mathrm{VS}}>0\) for all searchers included in the metric. Define the reward-to-contribution ratio
\(
r_i^{\mathcal M}=R_i^{\mathcal M}/\phi_i^{\mathrm{VS}}.
\)
We then measure searcher-side fairness by applying Jain's index to these ratios: a mechanism is fair when searchers receive similar rewards per unit of validator-dependent Shapley contribution.

\begin{definition}[Shapley-weighted Jain fairness]
The Shapley-weighted Jain fairness index of mechanism \(\mathcal M\) is
\[
\mathcal J_{\mathrm{Sh}}^{\mathcal M}
=
\frac{\left(\sum_{i=1}^n r_i^{\mathcal M}\right)^2}
{n\sum_{i=1}^n (r_i^{\mathcal M})^2}
=
\frac{\left(\sum_{i=1}^n R_i^{\mathcal M}/\phi_i^{\mathrm{VS}}\right)^2}
{n\sum_{i=1}^n (R_i^{\mathcal M}/\phi_i^{\mathrm{VS}})^2}.
\]
\end{definition}

The index satisfies \(1/n\le\mathcal J_{\mathrm{Sh}}^{\mathcal M}\le1\), provided that at least one searcher receives positive expected reward. It equals \(1\) in a fully contribution-fair allocation, where all searchers receive the same reward per unit of Shapley contribution; equivalently, \(R_i^{\mathcal M}=\lambda\phi_i^{\mathrm{VS}}\) for all \(i\) and some \(\lambda>0\). At the other extreme, it equals \(1/n\) when only one searcher has a positive reward-to-contribution ratio, meaning that rewards are maximally concentrated relative to contribution. Note Jain's original index measures equality among the quantities to which it is applied. Since we apply it to \(R_i^{\mathcal M}/\phi_i^{\mathrm{VS}}\), this metric measures equality of contribution-adjusted rewards.

The second economic metric measures how concentrated long-run expected rewards are among searchers. Whenever \(\sum_{j=1}^n R_j^M>0\), let
\(
s_i^{\mathcal M}=\frac{R_i^{\mathcal M}}{\sum_{j=1}^n R_j^{\mathcal M}}
\)
be searcher \(i\)'s expected reward share.

\begin{definition}[Effective reward decentralization]
The expected-reward Herfindahl--Hirschman Index (HHI) and effective number of rewarded searchers are
\[
\mathrm{ER\text{-}HHI}^{\mathcal M}=\sum_{i=1}^n (s_i^{\mathcal M})^2,
\qquad
N_{\mathrm{eff}}^{\mathcal M}=\frac{1}{\mathrm{ER\text{-}HHI}^{\mathcal M}}.
\]
\end{definition}

The Herfindahl--Hirschman Index (HHI) is a standard concentration measure. Here we apply it to expected reward shares rather than market shares, so \(\mathrm{ER\text{-}HHI}^{\mathcal M}\) measures concentration in searcher rewards. It ranges from \(1/n\), when all searchers receive equal expected reward shares, to \(1\), when one searcher receives all expected rewards. The reciprocal \(N_{\mathrm{eff}}^{\mathcal M}\) gives the effective number of rewarded searchers: it equals \(n\) under equal reward shares and equals \(1\) when rewards are fully concentrated on a single searcher. Unlike \(\mathcal J_{\mathrm{Sh}}^{\mathcal M}\), this metric does not adjust for contribution; it measures reward decentralization directly.

\subsubsection{Security properties}

The first security property concerns identity splitting by a searcher that can only copy or degrade its own execution code. For the security definitions, write \(\mathcal B_{-i}(o)=\bigcup_{j\neq i}\mathcal B_j(o)\). We use \(U_i^{\mathcal M}(\{b\}\cup\mathcal B_{-i}(o))\) to denote the payoff of real searcher \(i\) when it submits only the original submission \(b\), and \(U_i^{\mathcal M}(\sigma_i\cup\mathcal B_{-i}(o))\) to denote its payoff when it replaces \(b\) by a copied-code Sybil deviation \(\sigma_i\). The payoff \(U_i^{\mathcal M}\) aggregates over all submissions controlled by real searcher \(i\).

A strong benchmark is ex-post copied-code Sybil resistance~\cite{Pan2026SybilProof}. This property requires that copied-code identity splitting is never profitable for any realized profile of rival submissions.

\begin{definition}[Ex-post copied-code Sybil resistance]
A mechanism \(\mathcal M\) is ex-post copied-code Sybil-resistant if, for every opportunity \(o\), every real searcher \(i\), every original submission \(b\in\mathcal B_i(o)\), every realized rival submission set \(\mathcal B_{-i}(o)\), and every copied-code Sybil deviation \(\sigma_i\) from \(b\),
\(
U_i^{\mathcal M}(\{b\}\cup\mathcal B_{-i}(o))
\ge
U_i^{\mathcal M}(\sigma_i\cup\mathcal B_{-i}(o)).
\)
\end{definition}

This security requirement states that, once the rival submissions \(\mathcal B_{-i}(o)\) are fixed, a real searcher cannot increase its total payoff by replacing one original submission \(b\) with any feasible collection of copied-code submissions \(\sigma_i\).

We next introduce a weaker security requirement that evaluates copied-code deviations in expectation rather than state by state~\cite{Pan2026SybilProof}. Under this requirement, a copied-code Sybil deviation is ruled out if it does not increase the deviating searcher's expected payoff conditional on its own participation and certified normalized net value.

\begin{definition}[Bayesian copied-code Sybil resistance]
A mechanism \(\mathcal M\) is Bayesian copied-code Sybil-resistant if, for every real searcher \(i\), every \(z\in[\underline z_{i,\tau(o)},\overline z_{i,\tau(o)}]\), every original submission \(b\) with \(z_b=z\), and every copied-code Sybil deviation \(\sigma_i\) from \(b\),
\[
\mathbb E\!\left[
U_i^{\mathcal M}(\{b\}\cup\mathcal B_{-i}(o))
\mid D_{i,o}=1,\ z^\star_{i,o}=z
\right]
\ge
\mathbb E\!\left[
U_i^{\mathcal M}(\sigma_i\cup\mathcal B_{-i}(o))
\mid D_{i,o}=1,\ z^\star_{i,o}=z
\right].
\]
\end{definition}

This weaker Bayesian definition is appropriate for our introduced deadline-based system model (see Sect.~\ref{subsec:system-model}). Before the submission deadline, a deviating searcher observes its own execution code and certified normalized net value, but cannot condition its copied-code Sybil strategy on the realized submissions of its competitors.

The second security property extends copied-code Sybil resistance to validator--searcher coalitions. It requires that replacing an original submission with copied-code submissions cannot increase the combined payoff of the deviating searcher and the validator. This property is important because, if a copied-code deviation can increase joint surplus, the searcher and validator may have a direct incentive to coordinate the deviation and share the gain. As above, we define both an ex-post version and a Bayesian version of this requirement.

\begin{definition}[Ex-post validator--searcher coalition resistance]
A mechanism \(\mathcal M\) is ex-post validator--searcher coalition-resistant under copied-code deviations if, for every opportunity \(o\), every real searcher \(i\), every original submission \(b\in\mathcal B_i(o)\), every realized rival submission set \(\mathcal B_{-i}(o)\), and every feasible copied-code Sybil deviation \(\sigma_i\) from \(b\),
\[
\begin{aligned}
U_i^{\mathcal M}(\{b\}\cup\mathcal B_{-i}(o))
+
U_{\mathrm{val}}^{\mathcal M}(\{b\}\cup\mathcal B_{-i}(o))
\ge
U_i^{\mathcal M}(\sigma_i\cup\mathcal B_{-i}(o))
+
U_{\mathrm{val}}^{\mathcal M}(\sigma_i\cup\mathcal B_{-i}(o)).
\end{aligned}
\]
\end{definition}

This security requirement rules out copied-code deviations that increase the joint payoff of the deviating searcher and the validator, evaluated state by state for each realized rival submission set.

\begin{definition}[Bayesian validator--searcher coalition resistance]
A mechanism \(\mathcal M\) is Bayesian validator--searcher coalition-resistant under copied-code deviations if, for every real searcher \(i\), every \(z\in[\underline z_{i,\tau(o)},\overline z_{i,\tau(o)}]\), every original submission \(b\) with \(z_b=z\), and every feasible copied-code Sybil deviation \(\sigma_i\) from \(b\),
\[
\begin{aligned}
&\mathbb E\!\left[
U_i^{\mathcal M}(\{b\}\cup\mathcal B_{-i}(o))
+
U_{\mathrm{val}}^{\mathcal M}(\{b\}\cup\mathcal B_{-i}(o))
\mid D_{i,o}=1,\ z^\star_{i,o}=z
\right] \\
&\quad\ge
\mathbb E\!\left[
U_i^{\mathcal M}(\sigma_i\cup\mathcal B_{-i}(o))
+
U_{\mathrm{val}}^{\mathcal M}(\sigma_i\cup\mathcal B_{-i}(o))
\mid D_{i,o}=1,\ z^\star_{i,o}=z
\right].
\end{aligned}
\]
\end{definition}

The Bayesian version applies the same joint-payoff comparison from the interim perspective of the deviating searcher. It rules out copied-code deviations that increase the expected combined payoff of the searcher--validator coalition under the deadline-based information structure.

\section{First-Price Auction}
\label{sec:fpa}

We first analyze the sealed first-price auction (FPA), which is the natural benchmark for current private MEV bundle markets. In this mechanism, each active searcher submits a transfer offer to the validator; the highest offer wins, and the winner pays its own offer.

\subsection{FPA Mechanism}
For the FPA, the mechanism-specific message is the unnormalized transfer offer to the validator, \(m_b=\widetilde{\mathsf{bid}}_b\). We write
\(
\widetilde{\mathsf{bid}}_o
=
(\widetilde{\mathsf{bid}}_b)_{b\in\mathcal B(o)}
\)
for the vector of submitted transfer offers.
The first-price auction selects the submission with the highest transfer offer, breaking ties by a fixed priority order, such as encrypted-mempool arrival order, fixed before reveal\footnote{For copied-code deviations, copied submissions cannot obtain higher tie-breaking priority than the original submission they copy.}:
\[
\operatorname{Win}_o^{\mathrm{FPA}}
(\tilde z_o,\widetilde{\mathsf{bid}}_o)
\in
\arg\max_{b\in\mathcal B(o)}
\widetilde{\mathsf{bid}}_b.
\]
When the argmax contains multiple submissions, \(\operatorname{Win}_o^{\mathrm{FPA}}\) is the highest-priority submission among them.
If submission
\(
b^\star
=
\operatorname{Win}_o^{\mathrm{FPA}}
(\tilde z_o,\widetilde{\mathsf{bid}}_o)
\)
wins, the validator includes the execution code associated with \(b^\star\). The owner \(s(b^\star)\) receives payoff
\(
U_{s(b^\star)}^{\mathrm{FPA}}
(\tilde z_o,\widetilde{\mathsf{bid}}_o)
=
\tilde z_{b^\star}
-
\widetilde{\mathsf{bid}}_{b^\star},
\)
and the validator receives
\(
U_{\mathrm{val}}^{\mathrm{FPA}}
(\tilde z_o,\widetilde{\mathsf{bid}}_o)
=
\widetilde{\mathsf{bid}}_{b^\star}.
\)
All losing searchers receive zero.

For the FPA, each submission \(b\in\mathcal B_i(o)\) made by searcher \(i\) has normalized certified net value \(z_b\) and submits the normalized transfer offer
\(
\mathsf{bid}_b
=
\beta_{i,\tau(o)}(z_b),
\)
where \(\beta_{i,\tau(o)}\) is searcher \(i\)'s bidding function for opportunities of class \(\tau(o)\). The corresponding unnormalized transfer offer is
\(
\widetilde{\mathsf{bid}}_b
=
V_o\mathsf{bid}_b
=
V_o\beta_{i,\tau(o)}(z_b).
\)

The formal security analysis of the FPA is deferred to Appendix~\ref{appendix:Sec_FPA}. There we prove both copied-code Sybil resistance and validator--searcher coalition resistance in the strong ex-post sense; the corresponding Bayesian guarantees follow immediately.

\subsection{Expected Rewards and Economic Properties under FPA}
\label{subsec:fpa-economic}
This section characterizes the economic allocation induced by the FPA. The key feature is that rewards are assigned only to the submission with the highest transfer offer. Hence, under monotone bidding, the FPA rewards the event of being the highest-bidding active searcher rather than marginal contribution to the opportunity.

Let \(F_i=F_{i,\tau(o)}\) and \(\alpha_i=\alpha_{i,\tau(o)}\). Suppose first that the value distributions are continuous and the bidding functions \(\beta_{i,\tau(o)}\) are strictly increasing, so that, for any transfer offer \(\mathsf{bid}\) in the range of \(\beta_{i,\tau(o)}\), the inverse bidding function \(\beta_{i,\tau(o)}^{-1}(\mathsf{bid})\) is well defined. Conditional on searcher \(i\) being active and having normalized certified net value \(z\), its interim win probability is
\[
\Pi_i^{\mathrm{win}}(z)
=
\prod_{j\neq i}
\left(
1-\alpha_j+\alpha_j
F_j\!\left(
\beta_{j,\tau(o)}^{-1}(\beta_{i,\tau(o)}(z))
\right)
\right).
\]
Therefore searcher \(i\)'s expected reward under the FPA is
\[
R_i^{\mathrm{FPA}}
= \mathbb E[U_i^{\mathrm{FPA}}] = 
\alpha_i V_o
\int
\left(z-\beta_{i,\tau(o)}(z)\right)
\Pi_i^{\mathrm{win}}(z)\,dF_i(z),
\]
and the validator's expected payoff is
\[
R_{\mathrm{val}}^{\mathrm{FPA}}
= \mathbb E[U_\mathrm{val}^{\mathrm{FPA}}] =
V_o
\sum_{i=1}^n
\alpha_i
\int
\beta_{i,\tau(o)}(z)
\Pi_i^{\mathrm{win}}(z)\,dF_i(z).
\]
The economic metrics \(\mathcal J_{\mathrm{Sh}}^{\mathrm{FPA}}\), \(\mathrm{ER\text{-}HHI}^{\mathrm{FPA}}\), and \(N_{\mathrm{eff}}^{\mathrm{FPA}}\) are obtained by substituting \(R_i^{\mathrm{FPA}}\) into the definitions in Sect.~\ref{subsec:metrics}. To obtain sharper comparisons, we use the proportional-bidding benchmark
\(
\beta_{i,\tau(o)}(z)=\rho z
\)
for all searchers and opportunity classes, where \(\rho\in(0,1)\). Then the FPA ranks submissions by certified normalized net value, while the winner keeps fraction \(1-\rho\) of its certified net value. In this benchmark,
\[
R_i^{\mathrm{FPA}}
=
V_o(1-\rho)\Theta_i,
\qquad
\Theta_i
=
\alpha_i
\int
z
\prod_{j\neq i}
\left(1-\alpha_j+\alpha_jF_j(z)\right)
\,dF_i(z).
\]
Thus all searcher reward shares, and hence \(\mathrm{ER\text{-}HHI}^{\mathrm{FPA}}\) and \(N_{\mathrm{eff}}^{\mathrm{FPA}}\), are determined by the rank-dominance scores \((\Theta_i)_{i=1}^n\).

\begin{proposition}[Symmetric setting]
Suppose \(\alpha_i=\alpha\) and \(F_i=F\) for all \(i\). Under proportional bidding, \(R_i^{\mathrm{FPA}}=R_j^{\mathrm{FPA}}\) for all \(i,j\). Consequently,
\[
\mathrm{ER\text{-}HHI}^{\mathrm{FPA}}=\frac1n,
\qquad
N_{\mathrm{eff}}^{\mathrm{FPA}}=n.
\]
Moreover, since the Shapley contributions are also symmetric, \(\mathcal J_{\mathrm{Sh}}^{\mathrm{FPA}}=1\).
\end{proposition}

\begin{proof}
Symmetry gives \(\Theta_i=\Theta_j\) for all \(i,j\), hence equal expected rewards and equal reward shares. The HHI and effective-number claims follow immediately. The game \(W^{\mathrm{VS}}(\cdot)\) is symmetric across searchers, so \(\phi_i^{\mathrm{VS}}=\phi_j^{\mathrm{VS}}\) for all \(i,j\), and therefore \(R_i^{\mathrm{FPA}}/\phi_i^{\mathrm{VS}}\) is constant.
\end{proof}

The symmetric case provides a natural benchmark in which the FPA is decentralized and contribution-fair. Once searchers differ in coverage or execution quality, rewards become concentrated around the searchers most likely to win the rank competition.

\begin{proposition}[Dominant execution-quality class with heterogeneous coverage]
\label{prop:dominant-quality-class}
Suppose \(\alpha_i=\alpha_H\in(0,1)\) for every \(i\le m\) and \(\alpha_j=\alpha_L\in(0,1)\) for every \(j>m\), with \(\alpha_H\ge\alpha_L\). For a distribution \(F\), let \(\operatorname{supp}(F)\) denote the set of normalized net values attainable with positive probability. Under proportional bidding, suppose there exist \(m\in\{1,\ldots,n\}\) and \(\underline z_H,\overline z_L\in(0,1)\), with \(\underline z_H>\overline z_L\), such that \(F_i=F_H\) and \(\operatorname{supp}(F_H)\subseteq[\underline z_H,1]\) for every \(i\le m\), while \(\operatorname{supp}(F_j)\subseteq[0,\overline z_L]\) for every \(j>m\). Then, 
% writing \(p_H=1-(1-\alpha_H)^m\),
\[
N_{\mathrm{eff}}^{\mathrm{FPA}}
\le
m\left(
\frac{(1-(1-\alpha_H)^m)\underline z_H+(1-\alpha_H)^m(n-m)\alpha_L\overline z_L}
{(1-(1-\alpha_H)^m)\underline z_H}
\right)^2.
\]
Consequently, as \(\alpha_H\to1\), \(N_{\mathrm{eff}}^{\mathrm{FPA}}\to m\) and \(\mathcal J_{\mathrm{Sh}}^{\mathrm{FPA}}\to m/n\).
% if \(\phi_i^{\mathrm{VS}}>0\) for every searcher included in \(\mathcal J_{\mathrm{Sh}}^{\mathrm{FPA}}\), then \(\mathcal J_{\mathrm{Sh}}^{\mathrm{FPA}}\to m/n\).
\end{proposition}

\begin{proof}
Let \(p_H=1-(1-\alpha_H)^m\) be the probability that at least one high-quality searcher is active. Since \(\underline z_H>\overline z_L\), whenever at least one high-quality searcher is active, the FPA winner is in the high-quality class. Hence the total expected reward of the high-quality class is at least \(V_o(1-\rho)p_H\underline z_H\). A low-quality searcher can win only when no high-quality searcher is active, which occurs with probability \((1-\alpha_H)^m\). Since each low-quality searcher is active with probability \(\alpha_L\) and has normalized net value at most \(\overline z_L\), the total expected reward of the low-quality class is at most \(V_o(1-\rho)(1-\alpha_H)^m(n-m)\alpha_L\overline z_L\). Therefore, if \(s_H\) denotes the total reward share of the high-quality class, then
\[
s_H\ge
\frac{p_H\underline z_H}{p_H\underline z_H+(1-\alpha_H)^m(n-m)\alpha_L\overline z_L}.
\]
The \(m\) high-quality searchers are symmetric, since they have the same coverage probability and the same distribution \(F_H\). Thus they receive equal expected rewards, so each has reward share \(s_H/m\). Hence \(\mathrm{ER\text{-}HHI}^{\mathrm{FPA}}=\sum_i(s_i^{\mathrm{FPA}})^2\ge m(s_H/m)^2=s_H^2/m\). Since \(N_{\mathrm{eff}}^{\mathrm{FPA}}=1/\mathrm{ER\text{-}HHI}^{\mathrm{FPA}}\), we obtain
\[
N_{\mathrm{eff}}^{\mathrm{FPA}}
\le
m\left(
\frac{p_H\underline z_H+(1-\alpha_H)^m(n-m)\alpha_L\overline z_L}
{p_H\underline z_H}
\right)^2.
\]

As \(\alpha_H\to1\), we have \(p_H\to1\) and \((1-\alpha_H)^m\to0\), so the upper bound converges to \(m\). Low-quality reward shares converge to zero, while symmetry implies that the \(m\) high-quality searchers have equal limiting reward shares \(1/m\). Hence \(N_{\mathrm{eff}}^{\mathrm{FPA}}\to m\).

% Finally, assume \(\phi_i^{\mathrm{VS}}>0\) for every searcher included in \(\mathcal J_{\mathrm{Sh}}^{\mathrm{FPA}}\). As \(\alpha_H\to1\), the reward-to-contribution ratios of low-quality searchers converge to zero, while the \(m\) high-quality searchers have equal positive limiting ratios by symmetry. Therefore the Jain index is the Jain index of a vector with \(m\) equal positive components and \(n-m\) zero components, which equals \(m/n\). Thus \(\mathcal J_{\mathrm{Sh}}^{\mathrm{FPA}}\to m/n\).
It remains to prove the Jain-index limit. For any low-quality searcher \(j>m\), as \(\alpha_H\to1\),
\[
R_j^{\mathrm{FPA}}
\le
(1-\rho)V_o\alpha_L(1-\alpha_H)^m\overline z_L
\to0.
\]
On the other hand, using only the Shapley marginal contribution from the predecessor coalition \(\{\mathrm{val}\}\),
\[
\phi_j^{\mathrm{VS}}
\ge
\frac{1}{n(n+1)}
\Bigl(W^{\mathrm{VS}}(\{\mathrm{val},j\})-W^{\mathrm{VS}}(\{\mathrm{val}\})\Bigr)
=
\frac{V_o\alpha_L\mathbb E_{F_j}[z]}{n(n+1)}
>0.
\]
Hence \(R_j^{\mathrm{FPA}}/\phi_j^{\mathrm{VS}}\to0\) for every low-quality searcher \(j>m\). For the \(m\) high-quality searchers, symmetry implies equal expected rewards and equal validator--searcher Shapley contributions; their common reward-to-contribution ratio has a positive limit. Therefore the Jain vector converges to one with \(m\) equal positive components and \(n-m\) zero components. Its Jain index is
\(
\frac{(mc)^2}{nmc^2}=\frac{m}{n},
\)
for any \(c>0\). Thus \(\mathcal J_{\mathrm{Sh}}^{\mathrm{FPA}}\to m/n\).
\end{proof}

The proposition shows that FPA reward decentralization is governed by the size of the dominant class, not by the total number of searchers. When \(\alpha_H\) is high, at least one dominant searcher is almost always active, so lower-quality searchers almost never win. Thus, the FPA concentrates rewards on the top class: on \(m\) searchers in general, and on a single effective monopolist when \(m=1\).

\begin{corollary}[Single dominant searcher]
Under the assumptions of Proposition~\ref{prop:dominant-quality-class} with \(m=1\), if searcher \(1\) is the unique dominant searcher, then, as \(\alpha_H\to1\),
\(
N_{\mathrm{eff}}^{\mathrm{FPA}}\to1,
\quad
\mathcal J_{\mathrm{Sh}}^{\mathrm{FPA}}\to\frac1n.
\)
\end{corollary}

Figure~\ref{fig:fpa-multiple-elite-centralization} illustrates how increasing dominance affects reward concentration under the FPA. The horizontal axis varies the participation probability \(\alpha_H\) of the high-quality class, while different curves correspond to different numbers \(m_H\) of elite searchers. As \(\alpha_H\) increases, FPA rewards become increasingly concentrated within the dominant class: contribution-adjusted fairness declines toward \(m_H/n\), and the effective number of rewarded searchers declines toward \(m_H\). 

The results show that strong asymmetry can substantially reduce the economic decentralization of the FPA. With coverage heterogeneity, FPA rewards concentrate around searchers who access opportunities more often. With execution-quality heterogeneity, concentration is sharper: rewards are captured by the dominant execution class, and when this class contains only a few searchers, the FPA converges toward monopolistic concentration. Thus, the winner-take-all structure of the FPA can amplify centralization in realistic MEV environments where both opportunity access and execution quality are unequal.

\begin{figure}[t]
\centering
\includegraphics[width=\textwidth]{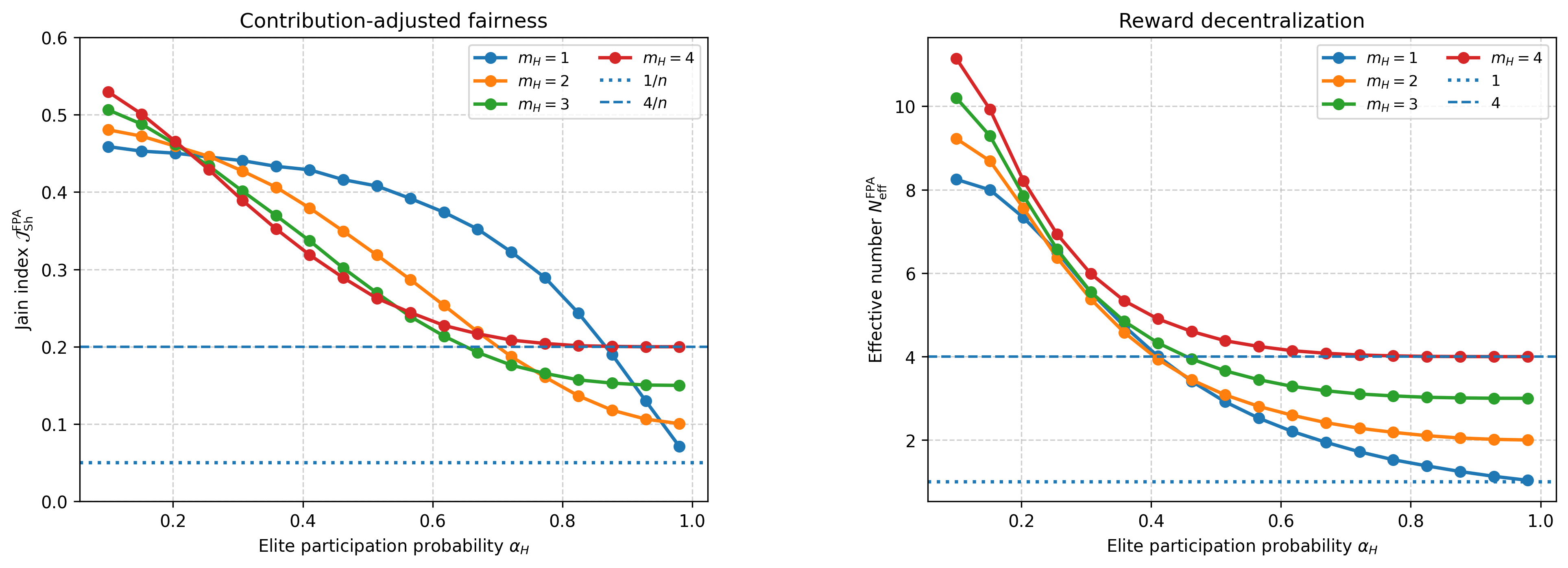}
\caption{Effect of dominant-class participation on FPA fairness and reward decentralization.}
\label{fig:fpa-multiple-elite-centralization}
\end{figure}

\section{The Entry-Filtered Shapley-Capped Auction}
\label{sec:shapley-mechanism}
This section introduces a Shapley-based searcher auction mechanism designed to distribute MEV rewards more fairly and broadly among competing searchers. The goal is to align rewards with searchers' marginal contributions, reduce winner-take-all concentration, and maintain security against copied-code Sybil deviations and validator--searcher coalition deviations in a permissionless MEV environment.

\subsection{Design Rationale}
\label{subsec:mechanism-rationale}

Our primary goal is to design a mechanism that, unlike a winner-take-all auction, distributes the searcher-side MEV reward among searchers that contribute to identifying or executing the same opportunity, rather than assigning all surplus to the top-ranked submission. The immediate difficulty is Sybil manipulation. Once rewards are distributed across multiple submissions, a searcher may submit copied or degraded versions of the same execution strategy to inflate its apparent contribution and capture a larger share. A natural solution is a Tullock-style contest~\cite{Tullock1980Rent}, where proportional allocation with costly bids can be Sybil-proof because multiple bids can be aggregated into a single equivalent bid \cite{Tullock1980Rent,Garimidi2026BeyondWTA}. However, this security relies on an all-pay structure: participants may obtain negative utility and must be willing to lose their bids. In an on-chain searcher auction, this would require deposits or enforceable losses, making the mechanism difficult to implement under current MEV infrastructure.

The broader obstruction is formal. Pan et al. show that any non-wasteful, symmetric, incentive-compatible, and ex-post Sybil-proof direct mechanism must reduce to a second-price auction with symmetric tie-breaking; equivalently, strong Sybil-proofness rules out nontrivial distributional objectives under those axioms \cite{Pan2026SybilProof}. Their Bayesian relaxation shows that this impossibility can be softened: an increasing-ticket-price mechanism can be Bayesian Sybil-proof by allocating the object through a lottery among threshold-clearing bidders whenever at least two such bidders exist, and otherwise falling back to a second-price auction \cite{Pan2026SybilProof}.
Nevertheless, this construction is not well matched to MEV searcher competition. Its Bayesian Sybil deterrence relies on an ex-ante ticket threshold dictated by an upper-tail cutoff: the point beyond which the bid distribution’s tail becomes no heavier than a uniform benchmark. While this condition is easily satisfied by distributions with flat or increasing right tails, it fails under the concentrated, declining upper tails typical of MEV signal models, such as the truncated log-normal distribution~\cite{adadurov2026open, DBLP:conf/icbc2/WuTLV24}. In these cases, the required cutoff is pushed to the maximum feasible valuation. Consequently, the ticket price also reaches this maximum, the lottery branch occurs with probability zero under continuous bids, and the mechanism entirely collapses back into a standard, winner-take-all second-price auction \cite{Pan2026SybilProof}. Appendix~\ref{app:pan-comparison} further discusses Pan et al.'s characterization and the reduction of their Bayesian mechanism to the second-price auction under our truncated-lognormal specification.

We therefore use the Shapley value as the economic allocation target, since it rewards searchers according to their marginal contribution to the opportunity. Shapley sharing alone, however, prevents neither copied-code Sybil manipulation nor validator--searcher collusion in a permissionless submission environment. The mechanism therefore combines Shapley-based rewards with residual caps and burn functions. The searcher-side burn is used to eliminate the expected gain from copied-code Sybil deviations, while the validator-side burn is calibrated to rule out profitable validator--searcher coalition deviations.

\subsection{Entry-Filtered Shapley-Capped Auction Mechanism}
\label{subsec:sca-mechanism}

The entry-filtered Shapley-capped auction mechanism is denoted by \(\mathrm{SCA}_\theta\), where \(\theta\) is the entry-filter threshold. Unlike the FPA, \(\mathrm{SCA}_\theta\) does not use transfer offers; its input is the set of valid submitted execution codes and their certified normalized net values. The mechanism first applies an entry filter: a submitted code is admitted only if its certified normalized net value is at least a protocol threshold \(\theta\in[0,1]\). The threshold is chosen ex ante to control the expected number of admitted submissions. For opportunity class \(\tau\), one may choose \(\theta=\theta_\tau\) so that \(\mathbb E[N_o^\theta]=\sum_i\alpha_{i,\tau}\Pr_{F_{i,\tau}}(z_{i,o}^\star\ge\theta)\le \bar N\), where \(\bar N\) is the target upper bound on the expected admitted population.

Fix an opportunity \(o\). Submissions with \(z_b<\theta\) are rejected and receive zero. Let \(\mathcal B^\theta(o)=\{b\in\mathcal B(o):z_b\ge\theta\}\) be the set of admitted submissions, and let \(t_\theta:=|\mathcal B^\theta(o)|\). Their normalized and unnormalized certified net-value profiles are denoted by \(z_o^\theta:=(z_b)_{b\in\mathcal B^\theta(o)}\) and \(\tilde z_o^\theta:=(\tilde z_b)_{b\in\mathcal B^\theta(o)}\), respectively. 
If \(t_\theta=0\), no submission is executed; set \(z_{(1)}^\theta:=0\), all searcher and validator payoffs to zero, and \(B_o^{\mathrm{SCA}_\theta}:=0\). Thus, \(\theta\) serves as a minimum certified-value requirement for execution.

Otherwise, order admitted submissions by \(z_{(1)}^\theta\ge\cdots\ge z_{(t_\theta)}^\theta\), where \(b_{(r)}^\theta\) is the rank-\(r\) admitted submission. Set \(z_{(t_\theta+1)}^\theta=0\) and \(\Delta_k^\theta=z_{(k)}^\theta-z_{(k+1)}^\theta\). The executed submission is \(\operatorname{Win}_o^{\mathrm{SCA}_\theta}(\tilde z_o^\theta)=b_{(1)}^\theta\), with ties broken by mempool arrival order.\footnote{If two submissions from different identities have the same certified value, the earlier submission in the encrypted mempool is ranked first. The ordering is fixed before decryption or reveal, so copied submissions cannot improve tie-breaking after observing rivals.}

For any coalition \(A \subseteq \mathcal B^\theta(o)\cup\{\mathrm{val}\}\), its realized value under \(\mathrm{SCA}_\theta\) is
\[
W_o^\theta(A)=
\begin{cases}
V_o\max_{b\in A\setminus\{\mathrm{val}\}} z_b,
& \mathrm{val}\in A,\ A\setminus\{\mathrm{val}\}\neq\varnothing,\\
0,& \text{otherwise}.
\end{cases}
\]
Note that the validator is necessary for a coalition to generate value.
For the admitted submissions \(\mathcal B^\theta(o)\), the Shapley values of the validator and the rank-\(r\) admitted submission are 
\[\Phi_{\mathrm{val}}^\theta(z_o^\theta)=V_o\sum_{k=1}^{t_\theta}\frac{k\Delta_k^\theta }{k+1}, \qquad \Phi_{(r)}^\theta(z_o^\theta)=V_o\sum_{k=r}^{t_\theta}\frac{\Delta_k^\theta}{k(k+1)}.\]
Hence \(\Phi_{\mathrm{val}}^\theta(z_o^\theta)+\sum_{r=1}^{t_\theta}\Phi_{(r)}^\theta(z_o^\theta)=V_oz_{(1)}^\theta\).

The mechanism includes two branches: the decentralized Shapley-sharing branch and the capped fallback branch. The mechanism places a cap \(H\ge2\) on the number of admitted submissions that can be handled in the decentralized Shapley-sharing branch. This cap is chosen relative to the admitted population, so \(H\) should be moderately larger than \(\mathbb E[N_o^\theta]\). Given retention parameters \(\kappa_S,\kappa_V\ge0\), define \(\eta_{t_\theta}^V=e^{-\kappa_V(\min\{t_\theta,H\}-1)}\) and \(\eta_{t_\theta}^S=e^{-\kappa_S(\min\{t_\theta,H\}-1)}\).
\begin{itemize}
    \item If \(t_\theta\le H\), all admitted submissions are rewarded according to retained Shapley values:
\[
U_{\mathrm{val}}^{\mathrm{SCA}_\theta}(\tilde z_o^\theta)
=
\eta_{t_\theta}^V\Phi_{\mathrm{val}}^\theta(z_o^\theta),
\qquad
P_{(r)}^{\mathrm{SCA}_\theta}(\tilde z_o^\theta)
=
\eta_{t_\theta}^S\Phi_{(r)}^\theta(z_o^\theta).
\]
\item If \(t_\theta>H\), the number of admitted submissions exceeds the cap, and the mechanism switches to the capped fallback branch:
\[
U_{\mathrm{val}}^{\mathrm{SCA}_\theta}(\tilde z_o^\theta)
=
\frac{H\eta_H^V}{H+1}V_oz_{(1)}^\theta,
\qquad
P_{(1)}^{\mathrm{SCA}_\theta}(\tilde z_o^\theta)
=
\frac{\eta_H^S}{H(t_\theta+1)}V_oz_{(1)}^\theta,
\]
and \(P_{(r)}^{\mathrm{SCA}_\theta}(\tilde z_o^\theta)=0\) for all \(r\ge2\).
\end{itemize}

Thus, when the number of admitted submissions is at most the cap, the mechanism uses Shapley sharing across admitted submissions. When the number of admitted submissions exceeds the cap, it switches to a fallback rule that pays only the top admitted submission on the searcher side while retaining a cap-dependent burn. For each real searcher \(i\), its total payoff is the sum of the payments assigned to its admitted submissions, \(U_i^{\mathrm{SCA}_\theta}(\tilde z_o^\theta)=\sum_{b\in\mathcal B_i(o)\cap\mathcal B^\theta(o)}P_b^{\mathrm{SCA}_\theta}(\tilde z_o^\theta)\). Any surplus remaining after paying the validator and all searchers is burned, \(B_o^{\mathrm{SCA}_\theta}(\tilde z_o^\theta)=V_oz_{(1)}^\theta-U_{\mathrm{val}}^{\mathrm{SCA}_\theta}(\tilde z_o^\theta)-\sum_iU_i^{\mathrm{SCA}_\theta}(\tilde z_o^\theta)\).

\begin{proposition}[Budget feasibility]
\label{prop:sca-theta-budget-feasibility}
For every opportunity and every valid submitted profile,
\[
U_{\mathrm{val}}^{\mathrm{SCA}_\theta}(\tilde z_o^\theta)
+\sum_{i\in\mathcal S_{\mathrm{act}}(o)}
U_i^{\mathrm{SCA}_\theta}(\tilde z_o^\theta)
+B_o^{\mathrm{SCA}_\theta}(\tilde z_o^\theta)
=
V_oz_{(1)}^\theta
\le V_o,
\qquad
B_o^{\mathrm{SCA}_\theta}(\tilde z_o^\theta)\ge0.
\]
\end{proposition}

\begin{proof}
If \(t_\theta=0\), the claim is immediate. If \(1\le t_\theta\le H\), Shapley budget balance gives \(\Phi_{\mathrm{val}}^\theta+\sum_{r=1}^{t_\theta}\Phi_{(r)}^\theta=V_oz_{(1)}^\theta\), and \(\eta_{t_\theta}^S,\eta_{t_\theta}^V\le1\) imply total payments at most \(V_oz_{(1)}^\theta\). If \(t_\theta>H\), total payments are at most
\[
\left(\frac{H}{H+1}+\frac{1}{H(t_\theta+1)}\right)V_oz_{(1)}^\theta
\le
\left(\frac{H}{H+1}+\frac{1}{H(H+2)}\right)V_oz_{(1)}^\theta
<
V_oz_{(1)}^\theta.
\]
Thus the residual defining \(B_o^{\mathrm{SCA}_\theta}\) is nonnegative, and the displayed identity follows by definition of the burn.
\end{proof}

The formal Bayesian security analysis of the entry-filtered SCA is given in Appendix~\ref{app:sca-theta-security}. It shows that it is sufficient to consider copy sizes \(q=2,\ldots,H+1\), while optimizing over all feasible copied-value profiles. The parameters \((\kappa_S,\kappa_V)\) are calibrated so that no such deviation increases the Bayesian interim payoff of either the searcher alone or the searcher--validator coalition. The calibration creates a security--efficiency tradeoff. Increasing \(\kappa_S\) strengthens searcher-side Sybil deterrence but also increases burn. Coalition resistance, by contrast, generally requires \(\kappa_V\) to lie within a feasible interval: identity-expansion deviations may impose a lower bound, while withdrawal deviations may impose an upper bound.
When this security region is nonempty, \((\kappa_S^\star,\kappa_V^\star)\) can be chosen to minimize expected burn subject to both Bayesian security requirements. The calibration reported in Appendix~\ref{app:revenue-decomposition} gives priority to minimizing searcher-side burn.

\subsection{Economic Properties of the Entry-Filtered SCA}
\label{subsec:sca-theta-economic}

We analyze the honest one-submission benchmark under the contribution-adjusted metrics defined in Sect.~\ref{subsec:metrics}. The entry filter affects rewards by excluding submissions below the threshold from the reward-sharing game. Thus, for a searcher \(i\), only states in which \(i\) is active and satisfies \(z_{i,o}^\star\ge\theta\) can generate a positive \(\mathrm{SCA}_\theta\) payment.

Let \(E_{i,o}^\theta:=D_{i,o}\mathbf 1\{z_{i,o}^\star\ge\theta\}\) indicate that searcher \(i\) is active and admitted, and let \(N_o^\theta:=\sum_iE_{i,o}^\theta\) be the admitted population. The cap \(H\) is chosen relative to \(N_o^\theta\), not the raw active count \(N_o\). In particular, the threshold \(\theta\) and the cap \(H\) should be calibrated jointly so that \(\mathbb E[N_o^\theta]\) is controlled and \(H\) is moderately larger than the typical admitted population.

For a realized profile, let \(\Phi_i^\theta(z_o^\theta)\) denote searcher \(i\)'s realized Shapley value in the admitted max game, with \(\Phi_i^\theta(z_o^\theta)=0\) when \(i\) is not admitted. Let \(\phi_i^{\mathrm{VS}}\) denote searcher \(i\)'s ex-ante Shapley contribution in the unfiltered validator--searcher game defined in Section~\ref{subsec:economic-metrics}. Define the admitted contribution coverage ratio
\(
a_{i,\theta}
:=
\frac{
\mathbb E[
\mathbf 1\{E_{i,o}^\theta=1,\ N_o^\theta\le H\}
\Phi_i^\theta(z_o^\theta)]
}{
\phi_i^{\mathrm{VS}}
},
\)
and the admitted contribution exposure ratio
\(
b_{i,\theta}
:=
\frac{
\mathbb E[
\mathbf 1\{E_{i,o}^\theta=1\}
\Phi_i^\theta(z_o^\theta)]
}{
\phi_i^{\mathrm{VS}}
}.
\)
The ratio \(a_{i,\theta}\) is the share of searcher \(i\)'s contribution denominator that is admitted and paid in the decentralized branch. The ratio \(b_{i,\theta}\) is an upper exposure bound for the contribution available to \(i\) after thresholding. These ratios are defined for all searchers and are normalized by \(\phi_i^{\mathrm{VS}}\). Unlike in the unfiltered case, \(b_{i,\theta}\) need not be at most one, because removing ineligible rivals can increase an admitted searcher's Shapley value relative to its full-profile Shapley value.

Define the global lower and upper reward-ratio bounds
\(
\lambda_\theta
:=
e^{-\kappa_S(H-1)}
\min_i a_{i,\theta}\) and
\(
\Lambda_\theta
:=
\max_i b_{i,\theta}.
\)
Thus \(\lambda_\theta\) and \(\Lambda_\theta\) are bounds on the reward-to-contribution ratios \(R_i^{\mathrm{SCA}_\theta}/\phi_i^{\mathrm{VS}}\) across all searchers.

\begin{proposition}[Entry-filtered SCA reward bounds]
\label{prop:sca-theta-reward-bounds}
Under honest one-submission play, assume the entry threshold is chosen so that the admitted population is bounded by \(t_{\max}\), and suppose the fallback searcher retention satisfies \(\eta_H^S\le H/t_{\max}\). Then, for every searcher \(i\),
\(
\lambda_\theta
\le
\frac{R_i^{\mathrm{SCA}_\theta}}{\phi_i^{\mathrm{VS}}}
\le
\Lambda_\theta.
\)
\end{proposition}

\begin{proof}
On the decentralized branch \(N_o^\theta\le H\), the entry-filtered SCA pays \(\eta_{N_o^\theta}^S\Phi_i^\theta(z_o^\theta)\). Since \(\eta_{N_o^\theta}^S=e^{-\kappa_S(N_o^\theta-1)}\ge e^{-\kappa_S(H-1)}\), we have
\[
R_i^{\mathrm{SCA}_\theta}
\ge
e^{-\kappa_S(H-1)}
\mathbb E[
\mathbf 1\{E_{i,o}^\theta=1,\ N_o^\theta\le H\}
\Phi_i^\theta(z_o^\theta)]
=
e^{-\kappa_S(H-1)}
a_{i,\theta}\phi_i^{\mathrm{VS}}.
\]
Taking the minimum over \(i\) gives the lower bound.

For the upper bound, \(U_i^{\mathrm{SCA}_\theta}\le\Phi_i^\theta(z_o^\theta)\) pointwise. This is immediate in the decentralized branch. In fallback, non-winners receive zero, and the winner receives
\[
\frac{\eta_H^S}{H(t_\theta+1)}V_oz_{(1)}^\theta
\le
\frac{1}{t_\theta(t_\theta+1)}V_oz_{(1)}^\theta
\le
\Phi_{(1)}^\theta(z_o^\theta),
\]
where the first inequality uses \(t_\theta\le t_{\max}\) and \(\eta_H^S\le H/t_{\max}\). Hence
\(
R_i^{\mathrm{SCA}_\theta}
\le
\mathbb E[
\mathbf 1\{E_{i,o}^\theta=1\}
\Phi_i^\theta(z_o^\theta)]
=
b_{i,\theta}\phi_i^{\mathrm{VS}}.
\)
Taking the maximum over \(i\) gives the upper bound.
\end{proof}

Define the contribution-side effective number by
\(
N_{\mathrm{eff}}^{\phi,\mathrm{VS}}
:=
\frac{
\left(\sum_{i=1}^n\phi_i^{\mathrm{VS}}\right)^2
}{
\sum_{i=1}^n\left(\phi_i^{\mathrm{VS}}\right)^2
},
\)
which measures how many economically meaningful searchers are represented by the validator--searcher Shapley contribution vector
\((\phi_i^{\mathrm{VS}})_{i=1}^n\).
\begin{theorem}[Economic performance of the entry-filtered SCA]
\label{thm:sca-theta-economic}
Under the assumptions of Proposition~\ref{prop:sca-theta-reward-bounds}, and assuming that \(\sum_{i=1}^n R_i^{\mathrm{SCA}_\theta}>0\),
\[
\mathcal J_{\mathrm{Sh}}^{\mathrm{SCA}_\theta}
\ge
\frac{4\lambda_\theta\Lambda_\theta}
{(\lambda_\theta+\Lambda_\theta)^2},
\qquad
N_{\mathrm{eff}}^{\mathrm{SCA}_\theta}
\ge
\left(\frac{\lambda_\theta}{\Lambda_\theta}\right)^2
N_{\mathrm{eff}}^{\phi,\mathrm{VS}}.
\]
\end{theorem}

\begin{proof}
% Let \(q_i:=R_i^{\mathrm{SCA}_\theta}/\phi_i^{\mathrm{VS}}\). By Proposition~\ref{prop:sca-theta-reward-bounds}, \(q_i\in[\lambda_\theta,\Lambda_\theta]\) for every \(i\). Jain's interval bound gives
% \[
% \mathcal J_{\mathrm{Sh}}^{\mathrm{SCA}_\theta}
% \ge
% \frac{4\lambda_\theta\Lambda_\theta}
% {(\lambda_\theta+\Lambda_\theta)^2}.
% \]
Let \(q_i:=R_i^{\mathrm{SCA}_\theta}/\phi_i^{\mathrm{VS}}\). By Proposition~\ref{prop:sca-theta-reward-bounds}, \(q_i\in[\lambda_\theta,\Lambda_\theta]\) for every \(i\). Therefore, for any given \(i\), 
\(
q_i^2
\le
(\lambda_\theta+\Lambda_\theta)q_i-\lambda_\theta\Lambda_\theta, \)
because \((q_i-\lambda_\theta)(q_i-\Lambda_\theta)\le0\). Summing over \(i\) gives
\(
\sum_i q_i^2
\le
(\lambda_\theta+\Lambda_\theta)\sum_i q_i
-
n\lambda_\theta\Lambda_\theta.
\)
For fixed bounds \([\lambda_\theta,\Lambda_\theta]\), the Jain index is minimized when the ratios are split between the two endpoints. Equivalently, the worst case satisfies
\(
\mathcal J(q)
=
\frac{(\sum_iq_i)^2}{n\sum_iq_i^2}
\ge
\frac{4\lambda_\theta\Lambda_\theta}
{(\lambda_\theta+\Lambda_\theta)^2}.
\)
Hence
\(
\mathcal J_{\mathrm{Sh}}^{\mathrm{SCA}_\theta}
\ge
\frac{4\lambda_\theta\Lambda_\theta}
{(\lambda_\theta+\Lambda_\theta)^2}.
\)
For the effective number, using \(R_i^{\mathrm{SCA}_\theta}=q_i\phi_i^{\mathrm{VS}}\),
\(
N_{\mathrm{eff}}^{\mathrm{SCA}_\theta}
=
\frac{(\sum_iq_i\phi_i^{\mathrm{VS}})^2}
{\sum_i(q_i\phi_i^{\mathrm{VS}})^2}
\ge
\frac{\lambda_\theta^2(\sum_i\phi_i^{\mathrm{VS}})^2}
{\Lambda_\theta^2\sum_i(\phi_i^{\mathrm{VS}})^2}
=
\left(\frac{\lambda_\theta}{\Lambda_\theta}\right)^2
N_{\mathrm{eff}}^{\phi,\mathrm{VS}}.
\)
\end{proof}

The entry threshold creates a security--coverage tradeoff. Increasing \(\theta\) reduces the expected admitted population,
\(
\mathbb E[N_o^\theta]
=
\sum_i\alpha_i\Pr(z_{i,o}^\star\ge\theta),
\)
which weakens copied-code Sybil incentives and can reduce the amount of searcher-side burn required for security. However, increasing \(\theta\) can reduce \(a_{i,\theta}\), because some genuine low-value contributions become ineligible for rewards. Thus the threshold changes how much of each searcher's contribution can receive a positive mechanism payment. For \(\theta>0\), the mechanism remains Shapley-based among admitted submissions, but avoids paying the entire long tail of low-certified-value submissions.

\subsection{Asymmetric Entry-Filtered Benchmark}
\label{sec:benchmark}

To evaluate the mechanisms under realistic searcher heterogeneity, we model execution costs by a truncated lognormal distribution. Conditional on participating in opportunity \(o\), searcher \(i\)'s normalized execution cost is
\[
C_i\mid D_{i,o}=1
\sim
\mathrm{TruncLogNormal}(\mu_i,\sigma^2;[\underline c_i,\overline c_i]),
\qquad
z_i^\star=1-C_i.
\]
Table~\ref{tab:benchmark-settings} introduced the four asymmetric searcher populations. Searchers differ in both participation probability and cost efficiency. Lower median cost corresponds to a higher certified normalized net value. Each benchmark scenario has three searcher classes. A class entry \(m(\alpha,\tilde c,[\underline c,\overline c])\) means \(m\) searchers with participation probability \(\alpha\), median cost \(\tilde c\), and truncated-lognormal support \([\underline c,\overline c]\). The four scenarios are ordered from less concentrated to more centralized competition.

\begin{table}[t]
\centering
\caption{Asymmetric searcher populations. Scenarios become more centralized from S1 to S4.}
\label{tab:benchmark-settings}
\resizebox{\textwidth}{!}{
\begin{tabular}{l c c l}
\hline
\textbf{Scenario}
& \(\boldsymbol n\)
& \(\boldsymbol{\mathbb E[N_o]}\)
& \textbf{Class structure } \(m(\alpha,\tilde c,[\underline c,\overline c])\) \\
\hline
S1: broad heterogeneity
& 20 & 9.60
& \(4(0.60,0.15,[0.10,0.25]),\ 8(0.50,0.30,[0.20,0.45]),\ 8(0.40,0.52,[0.35,0.70])\) \\

S2: aligned coverage--quality advantage
& 20 & 10.10
& \(5(0.85,0.14,[0.10,0.25]),\ 7(0.55,0.30,[0.20,0.45]),\ 8(0.25,0.55,[0.40,0.70])\) \\

S3: dominant high-quality class
& 20 & 8.00
& \(3(0.95,0.10,[0.05,0.20]),\ 7(0.45,0.30,[0.18,0.45]),\ 10(0.20,0.52,[0.35,0.70])\) \\

S4: single elite searcher
& 20 & 9.43
& \(1(0.98,0.10,[0.05,0.20]),\ 9(0.55,0.30,[0.18,0.45]),\ 10(0.35,0.52,[0.35,0.70])\) \\
\hline
\end{tabular}}
\end{table}

Table~\ref{tab:benchmark-results} compares the FPA with entry-filtered \(\mathrm{SCA}_\theta\) across different admission targets. For the entry-filtered \(\mathrm{SCA}_\theta\), we report three threshold levels. The threshold \(\theta\) is selected so that the expected admitted population satisfies \(\mathbb E[N_o^\theta]\in\{8,6,4\}\). The cap is then chosen relative to the admitted population, using \(H=\lceil \mathbb E[N_o^\theta]\rceil+3\). For each \(\theta\), \(\kappa_S^\star\) is the least searcher-side security parameter satisfying the Bayesian copied-code Sybil-resistance constraints over every realized value \(z\), every copy size \(q=2,\ldots,H+1\), and every feasible copied-value profile \(\mathbf z^c\in\mathcal Z_q^c(z)\). Notably, across all twelve benchmark calibrations, the binding deviation is the two-copy attack at the admission boundary, with \(q=2\), \(z=\theta\), and \(\mathbf z^c=(\theta,\theta)\). Hence, in these benchmark designs, the numerically worst deviation is a two-copy attack, although this need not hold in general.

\begin{table}[t]
\centering
\caption{FPA versus entry-filtered SCA. Within each scenario, the SCA rows show increasingly selective admission rules; the FPA row is reported once as the winner-take-all benchmark.}
\label{tab:benchmark-results}
\resizebox{0.8\textwidth}{!}{
\begin{tabular}{l l c c c c c c c}
\hline
\textbf{Scenario}
& \textbf{Mechanism}
& \(\boldsymbol{\mathbb E[N_o^\theta]}\)
& \(\boldsymbol{\theta}\)
& \(\boldsymbol H\)
& \(\boldsymbol{\kappa_S^\star}\)
& \(\boldsymbol{\mathcal J_{\mathrm{Sh}}}\)
& \(\boldsymbol{N_{\mathrm{eff}}}\)
& \textbf{Note} \\
\hline

S1
& \(\mathrm{SCA}_\theta\)
& 8 & 0.495 & 11 & 0.371
& \cellcolor{green!25}0.948
& \cellcolor{green!25}11.10
& broad admission \\

& \(\mathrm{SCA}_\theta\)
& 6 & 0.620 & 9 & 0.293
& \cellcolor{yellow!25}0.706
& \cellcolor{yellow!25}9.67
& intermediate admission \\

& \(\mathrm{SCA}_\theta\)
& 4 & 0.719 & 7 & 0.165
& \cellcolor{orange!20}0.574
& \cellcolor{orange!20}7.19
& selective admission \\

& FPA
& -- & -- & -- & --
& \cellcolor{red!18}0.244
& \cellcolor{red!18}4.26
& winner-take-all benchmark \\
\hline

S2
& \(\mathrm{SCA}_\theta\)
& 8 & 0.580 & 11 & 0.415
& \cellcolor{green!25}0.661
& \cellcolor{green!25}8.72
& broad admission \\

& \(\mathrm{SCA}_\theta\)
& 6 & 0.709 & 9 & 0.356
& \cellcolor{yellow!25}0.536
& \cellcolor{yellow!25}6.96
& intermediate admission \\

& \(\mathrm{SCA}_\theta\)
& 4 & 0.790 & 7 & 0.265
& \cellcolor{orange!20}0.288
& \cellcolor{orange!20}5.21
& selective admission \\

& FPA
& -- & -- & -- & --
& \cellcolor{red!18}0.251
& \cellcolor{red!18}5.00
& winner-take-all benchmark \\
\hline

S3
& \(\mathrm{SCA}_\theta\)
& 8 & 0.300 & 11 & 0.406
& \cellcolor{green!25}0.996
& \cellcolor{green!25}6.07
& broad admission \\

& \(\mathrm{SCA}_\theta\)
& 6 & 0.598 & 9 & 0.355
& \cellcolor{yellow!25}0.693
& \cellcolor{yellow!25}5.45
& intermediate admission \\

& \(\mathrm{SCA}_\theta\)
& 4 & 0.734 & 7 & 0.284
& \cellcolor{orange!20}0.413
& \cellcolor{orange!20}4.01
& selective admission \\

& FPA
& -- & -- & -- & --
& \cellcolor{red!18}0.150
& \cellcolor{red!18}3.00
& winner-take-all benchmark \\
\hline

S4
& \(\mathrm{SCA}_\theta\)
& 8 & 0.467 & 11 & 0.388
& \cellcolor{green!25}0.970
& \cellcolor{green!25}3.84
& broad admission \\

& \(\mathrm{SCA}_\theta\)
& 6 & 0.597 & 9 & 0.322
& \cellcolor{yellow!25}0.726
& \cellcolor{yellow!25}3.78
& intermediate admission \\

& \(\mathrm{SCA}_\theta\)
& 4 & 0.683 & 7 & 0.201
& \cellcolor{orange!20}0.500
& \cellcolor{orange!20}3.23
& selective admission \\

& FPA
& -- & -- & -- & --
& \cellcolor{red!18}0.075
& \cellcolor{red!18}1.04
& winner-take-all benchmark \\
\hline
\end{tabular}}
\end{table}

Table~\ref{tab:benchmark-results} supports the main decentralization claim. As competition becomes more concentrated, the FPA reward allocation collapses: \(N_{\mathrm{eff}}^{\mathrm{FPA}}\) falls from \(4.26\) in S1 to \(1.04\) in S4, and \(\mathcal J_{\mathrm{Sh}}^{\mathrm{FPA}}\) falls from \(0.244\) to \(0.075\). By contrast, the entry-filtered \(\mathrm{SCA}_\theta\) consistently preserves a broader reward base. In the most centralized case, S4, the FPA is effectively controlled by one elite searcher, while \(\mathrm{SCA}_\theta\) with \(\mathbb E[N_o^\theta]=6\) achieves \(N_{\mathrm{eff}}^{\mathrm{SCA}_\theta}=3.78\) and \(\mathcal J_{\mathrm{Sh}}^{\mathrm{SCA}_\theta}=0.726\).

The other side of the result is the security--coverage tradeoff. Lower thresholds admit more searchers and improve decentralization, but require larger \(\kappa_S^\star\). Higher thresholds reduce \(\kappa_S^\star\), but exclude more low-value submissions from reward sharing. Thus, relative to the FPA, the entry-filtered \(\mathrm{SCA}_\theta\) offers substantially better reward decentralization, with the threshold \(\theta\) controlling the cost of Bayesian Sybil resistance.

Additional numerical results on increasing elite execution advantage, together with the payment decomposition and validator-side security calibration, are reported in Appendix~\ref{app:revenue-decomposition}.

\section{Empirical Benchmarks}
We evaluate the performance of searchers by analyzing historical MEV transactions on Ethereum\footnote{Our implementation is available at {\url{https://gitlab.esat.kuleuven.be/mev/searcher_game}}.}. 
In practice, a searcher entity may deploy multiple contracts to identify and capture MEV opportunities. However, establishing correlations between different contracts controlled by the same searcher entity is outside the scope of this paper. 
For the purposes of this section, the terms searcher and contract are treated as interchangeable.

\smallskip\noindent\textbf{Dataset.}
We use a dataset~\cite{DBLP:conf/ndss/LuoLLHLMSC26} developed by Luo et al., which contains 577,264 arbitrage transactions from block 16308203 to 18908892. 
In total, there are 855 contracts generating a total revenue of 37.79 million dollars and a total profit of 8.88 million dollars.

\smallskip\noindent\textbf{Searcher strategies.}
Arbitrage searchers implement one or more trading strategies to capture MEV opportunities. 
These strategies often involve constructing cyclic transactions across decentralized exchanges, in which two or more tokens are swapped within a single atomic execution. 
When multiple searchers identify the same opportunity, they compete for inclusion by submitting higher bids.

We identify two commonly implemented arbitrage strategies.
One prioritizes selectivity, targeting scarce opportunities leading to high profit margins. 
The other prioritizes volume, executing a large number of transactions with relatively low profit per transaction. 
%This distinction mirrors patterns observed among arbitrageurs in traditional financial markets, suggesting that strategy choice is closely tied to the type of opportunity covered. %, as discussed in~\cite{}.

The contracts in the dataset are therefore classified into three strategy profiles: (C1) margin-oriented, which contains contracts in the top 10\% by average profit per transaction; (C2) volume-oriented, which are those in the top 10\% by total transaction count; and (C3) mixed-profile contracts, whose observed performance does not show a dominant margin-oriented or volume-oriented profile. 

Among the classified contracts, ten satisfy both the margin-oriented and volume-oriented criteria. 
We assign these overlapping contracts to the volume-oriented group, since high transaction volume more directly captures repeated execution behavior. 
In addition, two contracts \texttt{0x03c6...67cc} and \texttt{0xeeaa...3376} are excluded from margin-oriented group after identifying volume-oriented traits in their behavior.

A summary of the classified contracts is shown in Table~\ref{tab:strategy-basic-summary}. 
The profit retained by a searcher varies substantially depending on the type of opportunity covered, with the average profit per transaction in the C1 group being approximately 30 times that of the C2 group.
\begin{table}[t]
\centering
\caption{Summary of Searcher Strategies.}
\label{tab:strategy-basic-summary}
\resizebox{0.8\textwidth}{!}{
\begin{tabular}{lrrrrr}
\hline
Strategy & \#Contracts & \#Tx & Revenue Share & Profit Share & Avg. Profit / Tx \\
\hline
C1 & 74  & 3,323   & 4.35\%  & 14.02\% & \$374.74 \\
C2 & 88  & 510,620 & 76.29\% & 73.57\% & \$12.79 \\
C3 & 693 & 63,321  & 19.36\% & 12.40\% & \$17.39 \\
\hline
\end{tabular}
}
\end{table}

\smallskip\noindent\textbf{Execution efficiency across strategies.}
Within each strategy group, contracts are further classified into three tiers based on their contribution to the group's total realized profit. 
Contracts are first ranked in descending order by their total profit within the strategy. 
We then compute the cumulative share of positive realized profit accounted for by this ranking. 
Elite contracts are defined as the smallest set of top-ranked contracts whose cumulative positive profit reaches 60\% of the strategy’s total positive profit. 
Middle contracts are the next set of ranked contracts needed to raise the cumulative positive-profit share from 60\% to 90\%.
Tail contracts are all remaining contracts, including contracts with zero or negative realized profit.
For each tier, we show the median profit-to-revenue ratio, which measures the share of generated revenue retained by searchers. 
A higher profit-to-revenue ratio indicates that searchers retain more value after paying the costs required to win inclusion.

The performance of each tier is summarized in Table~\ref{tab:strategy-tier-summary}. 
A strong concentration is observed within each strategy group, with a small fraction of top-tier contracts accounting for over 60\% of total realized profit.
A stark contrast emerges between the two strategies: the bottom tier of C1 retains a higher median profit-to-revenue ratio than the top tier of C2, 35.66\% and 17.67\% respectively. 
This hints that C1 faces considerably lower competition intensity.

% \begin{table}[t]
% \centering
% \caption{Summary of contract tiers within each strategy. (M) denotes median. }
% \label{tab:strategy-tier-summary}
% \begin{tabular}{lrrrrr}
% \hline
% Tier & \#Contracts & \#Tx & Profit/Tx (M) & Profit Share & Profit/Revenue (M)\\
% \hline
% C1-elite  & 9   & 1,950     & \$561.76 & 62.08\% & 71.15\% \\
% C1-middle & 17  & 988     & \$491.89 & 28.15\% & 40.28\% \\
% C1-tail   & 48  & 385      & \$274.35 & 9.77\%  & 35.66\% \\\hline
% C2-elite  & 11  & 53,719  & \$128.11 & 61.65\% & 17.67\% \\
% C2-middle & 19  & 99,565   & \$38.32  & 28.86\% & 5.41\% \\
% C2-tail   & 58  & 357,336   & \$1.04   & 9.48\%  & 1.26\% \\\hline
% C3-elite  & 26  & 14,836     & \$45.37  & 62.12\% & 27.25\% \\
% C3-middle & 51  & 15,782     & \$23.47  & 30.68\% & 18.74\% \\
% C3-tail   & 616 & 32,703       & \$0.02   & 7.20\%  & 1.11\% \\
% \hline
% \end{tabular}
% \end{table}
\begin{table}[b]
\centering
\caption{Summary of contract tiers within each strategy. (M) denotes median. }
\label{tab:strategy-tier-summary}
{\small
\resizebox{0.8\textwidth}{!}{
\begin{tabular}{lrrrrr}
\hline
Tier & \#Contracts & \#Tx & Profit/Tx (M) & Profit Share & Profit/Revenue (M)\\
\hline
C1-elite  & 9   & 1,950     & \$561.76 & 62.08\% & 71.15\% \\
C1-middle & 17  & 988     & \$491.89 & 28.15\% & 40.28\% \\
C1-tail   & 48  & 385      & \$274.35 & 9.77\%  & 35.66\% \\\hline
C2-elite  & 11  & 53,719  & \$128.11 & 61.65\% & 17.67\% \\
C2-middle & 19  & 99,565   & \$38.32  & 28.86\% & 5.41\% \\
C2-tail   & 58  & 357,336   & \$1.04   & 9.48\%  & 1.26\% \\\hline
C3-elite  & 26  & 14,836     & \$45.37  & 62.12\% & 27.25\% \\
C3-middle & 51  & 15,782     & \$23.47  & 30.68\% & 18.74\% \\
C3-tail   & 616 & 32,703       & \$0.02   & 7.20\%  & 1.11\% \\
\hline
\end{tabular}
}
}
\end{table}

Figure~\ref{fig:S1-S2-profit-revenue-ration-distribution} illustrates the transaction distribution across profit-to-revenue ratio ranges. 
The profit-to-revenue ratio measures the fraction of arbitrage revenue that remains as realized profit after fees and bids, and thus provides an observable proxy for the retained margin of included transactions. 
Within each strategy, elite contracts generally exhibit higher median profit-to-revenue ratios than middle and tail contracts, indicating that the contracts accounting for the largest share of realized profit also tend to retain a larger fraction of revenue on a per-transaction basis. 
Across strategies, C1 displays substantially higher retained margins than C2 and C3. 
This pattern is consistent with the interpretation that C1 searchers target arbitrage opportunities with lower relative execution or bidding costs, whereas C2 searchers operate in a more cost-intensive environment associated with high-frequency opportunity capture. 
Since the dataset contains only successful transactions and does not observe failed bids, the profit-to-revenue ratio should be interpreted as an ex-post measure of realized margin rather than a direct measure of auction-level bidding power.

\begin{figure}[t]
\centering
\includegraphics[scale=0.3]{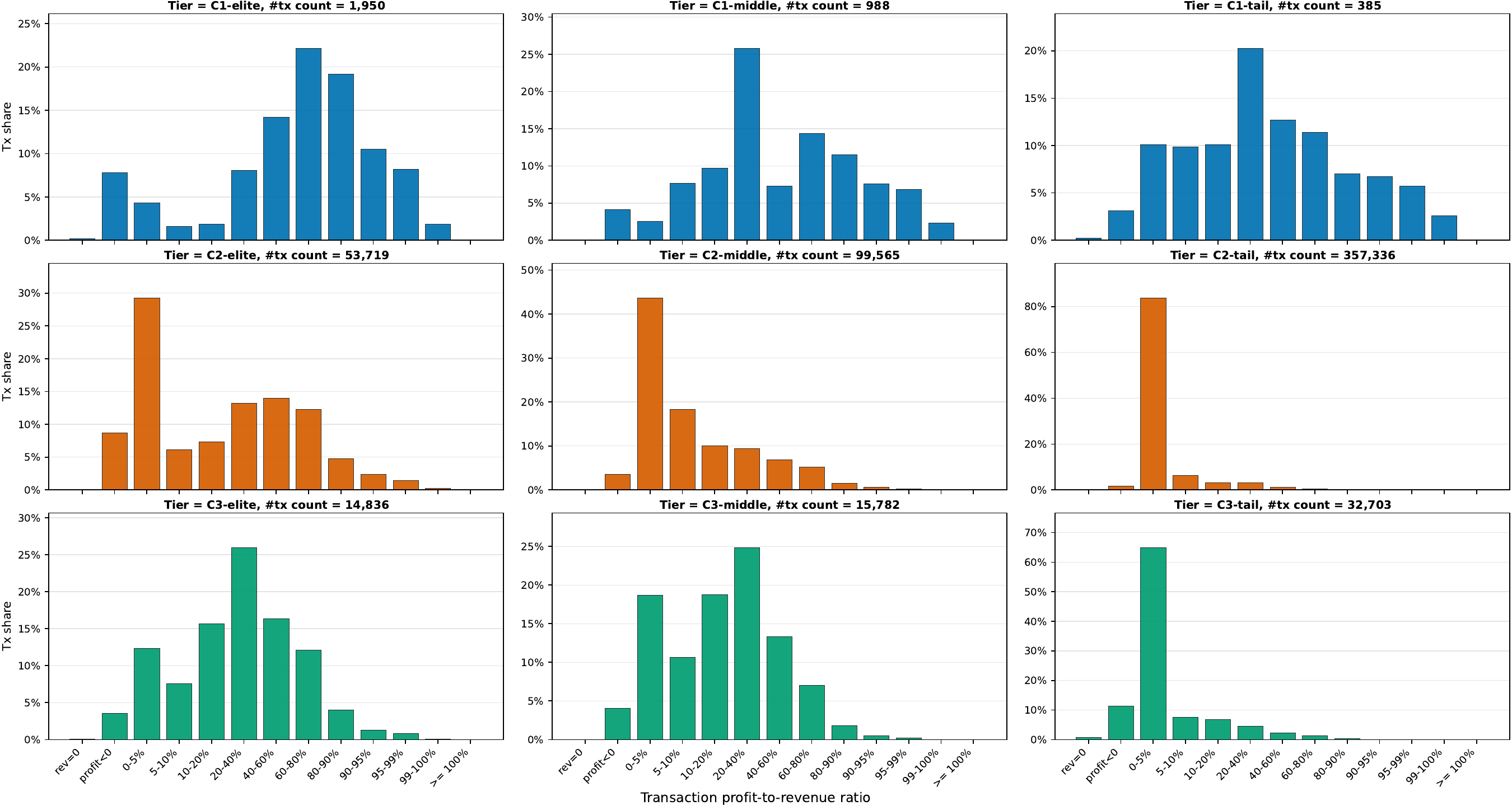}
\caption{Transaction distribution over profit-to-revenue ratio bins, grouped by profit tier. 
The x-axis denotes profit-to-revenue ratio bins, while the y-axis reports the percentage share of transactions.}
\label{fig:S1-S2-profit-revenue-ration-distribution}
\end{figure}

\smallskip\noindent\textbf{Empirical calibration and limitations.}
Table~\ref{tab:empirical-extracted-distributions} reports the truncated-lognormal cost distributions extracted from the Ethereum arbitrage data. The opportunity types \(C_1\) and \(C_2\) correspond to different empirical arbitrage environments: \(C_1\) is margin-oriented, while \(C_2\) is volume-oriented. Within each opportunity type, contracts are grouped into elite, middle, and tail tiers. A class entry \(m(\alpha,\tilde c,[\underline c,\overline c])\) represents \(m\) contracts with relative coverage \(\alpha\), median normalized cost \(\tilde c\), and truncated-lognormal support \([\underline c,\overline c]\).

A limitation is that the dataset contains only successful on-chain transactions. We therefore observe included winners, but not the full set of searchers competing for each opportunity, unsuccessful submissions, or builder- and relay-side auction messages. Since the model is defined over the full opportunity-level competition environment, the empirical calibration should be interpreted as a contract-level proxy rather than a fully identified structural estimate. Builder, relay, or private order-flow data would permit a more precise calibration of participation probabilities and value distributions.

\begin{table}[b]
\centering
\caption{Empirically extracted truncated-lognormal cost distributions.}
\label{tab:empirical-extracted-distributions}

\small
\setlength{\tabcolsep}{4pt}
\renewcommand{\arraystretch}{0.85}

\begin{tabular}{l l c c c}
\hline
\textbf{Opportunity type}
& \textbf{Tier}
& \(\boldsymbol m\)
& \(\boldsymbol\alpha\)
& \(\boldsymbol{\tilde c,[\underline c,\overline c]}\) \\
\hline
\(C_1\) & Elite  & 9  & 0.950 & \(0.2885,[0.01,0.95]\) \\
\(C_1\) & Middle & 17 & 0.255 & \(0.5972,[0.05,1.00]\) \\
\(C_1\) & Tail   & 48 & 0.035 & \(0.6434,[0.01,1.00]\) \\
\hline
\(C_2\) & Elite  & 11 & 0.753 & \(0.8233,[0.01,1.00]\) \\
\(C_2\) & Middle & 19 & 0.808 & \(0.9459,[0.10,1.00]\) \\
\(C_2\) & Tail   & 58 & 0.950 & \(0.9874,[0.20,1.00]\) \\
\hline
\end{tabular}
\end{table}

Table~\ref{tab:empirical-fpa-sca-n14} compares the FPA and the entry-filtered \(\mathrm{SCA}_\theta\) under these empirical settings. The raw expected active populations are \(\mathbb E[N_o]=14.565\) for \(C_1\) and \(\mathbb E[N_o]=78.735\) for \(C_2\). For each opportunity type, \(\theta\) is selected so that \(\mathbb E[N_o^\theta]=14\), and the cap is set to \(H=17\). Thus, the target corresponds to near-full admission in \(C_1\), but selective admission in \(C_2\).

For each \(\theta\), \(\kappa_S^\star\) is the smallest searcher-side parameter satisfying the Bayesian copied-code Sybil-resistance inequalities for every participating class, every realized value \(z\ge\theta\), every \(q=2,\ldots,H+1\), and every feasible copied-value profile \(\mathbf z^c\in\mathcal Z_q^c(z)\). The copied-value profile is selected before the rival profile is realized. In both empirical calibrations, the numerical supremum is attained by a two-copy exact-copy deviation: for \(C_1\), by a middle-tier contract at \(z\approx0.143>\theta\), and for \(C_2\), by a middle-tier contract at \(z=\theta\).

The empirical comparison suggests that \(\mathrm{SCA}_\theta\) is most valuable when FPA rewards are relatively concentrated. This is the case for \(C_1\): the contribution-adjusted Jain index increases from \(0.175\) under the FPA to \(0.983\) under \(\mathrm{SCA}_\theta\), while \(N_{\mathrm{eff}}\) increases from \(9.10\) to \(11.11\). The increase in \(N_{\mathrm{eff}}\) is more modest because the calibration treats the nine elite \(C_1\) contracts as distinct searchers with identical class-level primitives. In practice, these contracts may be heterogeneous or share common ownership, implying stronger FPA reward concentration.

For \(C_2\), the FPA is already highly decentralized at the contract level, with \(\mathcal J_{\mathrm{Sh}}^{\mathrm{FPA}}=0.997\) and \(N_{\mathrm{eff}}^{\mathrm{FPA}}=79.54\). The entry-filtered SCA therefore produces little additional contribution-adjusted fairness, but increases the effective reward base to \(N_{\mathrm{eff}}^{\mathrm{SCA}_\theta}=84.05\) while satisfying the reported Bayesian copied-code Sybil-resistance constraints.
Overall, the empirical evidence supports the interpretation that the entry-filtered SCA is most consequential in relatively concentrated opportunity classes such as \(C_1\). In already diffuse classes such as \(C_2\), its economic effect is smaller, and its primary role is to preserve broad reward allocation under an explicit searcher-side Bayesian security calibration.

\begin{table}[t!]
\centering
\caption{FPA versus entry-filtered SCA for the empirical Ethereum-calibrated settings with \(\mathbb E[N_o^\theta]=14\). The SCA parameter \(\kappa_S^\star\) is calibrated using the full Bayesian copied-value supremum.}
\label{tab:empirical-fpa-sca-n14}
\resizebox{\textwidth}{!}{
\begin{tabular}{l l c c c c c c c}
\hline
\textbf{Opportunity type}
& \textbf{Mechanism}
& \(\boldsymbol{\mathbb E[N_o^\theta]}\)
& \(\boldsymbol{\theta}\)
& \(\boldsymbol H\)
& \(\boldsymbol{\kappa_S^\star}\)
& \(\boldsymbol{\mathcal J_{\mathrm{Sh}}}\)
& \(\boldsymbol{N_{\mathrm{eff}}}\)
& \textbf{Note} \\
\hline
\(C_1\)
& \(\mathrm{SCA}_\theta\)
& 14 & 0.141 & 17 & 0.527
& \cellcolor{green!25}0.983
& \cellcolor{green!25}11.11
& near-full admission \\

& FPA
& -- & -- & -- & --
& \cellcolor{orange!18}0.175
& \cellcolor{orange!18}9.10
& winner-take-all benchmark \\
\hline
\(C_2\)
& \(\mathrm{SCA}_\theta\)
& 14 & 0.392 & 17 & 0.457
& \cellcolor{green!25}0.999
& \cellcolor{green!25}84.05
& selective admission \\

& FPA
& -- & -- & -- & --
& \cellcolor{orange!18}0.997
& \cellcolor{orange!18}79.54
& winner-take-all benchmark \\
\hline
\end{tabular}}
\end{table}

\section{Conclusion}
\label{sec:conclusion}

This paper argues that searcher competition is a distinct and important locus of MEV centralization. While the current first-price, winner-take-all auction is robust against copied-code identity splitting, its economic allocation is inherently rank-based: rewards accrue to the searcher most likely to submit the highest-value execution, rather than to searchers in proportion to their marginal contribution. Under heterogeneous access and execution quality, rank dominance can amplify small advantages into persistent reward concentration.

Motivated by this limitation, we introduced the entry-filtered Shapley-capped auction as one possible alternative. The mechanism preserves efficient execution by selecting the highest certified-value submission, but reallocates part of the retained surplus according to contribution among admitted high-quality submissions. Its cap, entry filter, and burn parameters expose the central design tradeoff: broader reward sharing improves contribution-adjusted fairness and reward decentralization, but must be disciplined by explicit Bayesian security constraints against copied-code Sybil deviations and validator--searcher coalitions. Overall, the results show that decentralization in MEV markets is not only a matter of infrastructure or order-flow access but also a mechanism-design problem requiring fairness, efficiency, and permissionless security to be addressed jointly.

%%
%% Bibliography
%%

%% Please use bibtex, 

\bibliography{lipics-v2021-sample-article}

\appendix

\section{Security Analysis of the FPA}
\label{appendix:Sec_FPA}
For this section, we impose the following assumptions on the FPA bidding functions and tie-breaking rule. For every real searcher \(i\), all submissions controlled by \(i\) for opportunities of class \(\tau(o)\), including copied-code submissions, use the same searcher-specific bidding function \(\beta_{i,\tau(o)}\). The bidding function is feasible and monotone: \(0\le \beta_{i,\tau(o)}(z)\le z\), \(\beta_{i,\tau(o)}\) is non-decreasing, and the retained surplus \(z-\beta_{i,\tau(o)}(z)\) is non-decreasing in \(z\). Ties are broken by the fixed priority order specified in the FPA mechanism, and copied-code submissions cannot obtain higher tie-breaking priority than the original submission they copy.

\begin{theorem}[Ex-post copied-code Sybil resistance of the FPA]
Under the assumptions above, the FPA is ex-post copied-code Sybil-resistant.
\end{theorem}

\begin{proof}
Fix an opportunity \(o\), a real searcher \(i\), an original submission \(b\in\mathcal B_i(o)\), and a realized rival submission set \(\mathcal B_{-i}(o)\). Consider any copied-code Sybil deviation \(\sigma_i=(\hat b_1,\ldots,\hat b_k)\) from \(b\). By definition, \(z_{\hat b_\ell}\le z_b \) for every \(\ell=1,\ldots,k.\)
Since all submissions controlled by \(i\) use the same non-decreasing bidding function \(\beta_{i,\tau(o)}\),
\[
\beta_{i,\tau(o)}(z_{\hat b_\ell})
\le
\beta_{i,\tau(o)}(z_b)
\quad
\text{for every } \ell=1,\ldots,k.
\]
Thus no copied submission can submit a higher transfer offer than the original submission. Moreover, if a copied submission submits the same transfer offer as \(b\), it cannot have higher tie-breaking priority than \(b\).

If \(b\) loses, then every copied submission also loses, and the deviating searcher receives zero. If \(b\) wins, any winning copied submission \(\hat b_\ell\) satisfies \(z_{\hat b_\ell}\le z_b\). Since \(z-\beta_{i,\tau(o)}(z)\) is non-decreasing,
\(
z_{\hat b_\ell}-\beta_{i,\tau(o)}(z_{\hat b_\ell})
\le
z_b-\beta_{i,\tau(o)}(z_b).
\)
Therefore the deviating searcher's total payoff cannot exceed its payoff from \(b\). Hence,
\[
U_i^{\mathrm{FPA}}(\{b\}\cup\mathcal B_{-i}(o))
\ge
U_i^{\mathrm{FPA}}(\sigma_i\cup\mathcal B_{-i}(o)).
\]
\end{proof}

Since ex-post copied-code Sybil resistance holds state by state, the Bayesian version follows immediately.

\begin{corollary}[Bayesian copied-code Sybil resistance of the FPA]
Under the assumptions of the preceding theorem, the FPA is Bayesian copied-code Sybil-resistant.
\end{corollary}

We now consider validator--searcher coalition deviations under the same copied-code restriction. In the FPA, the winning transfer is paid by the searcher and received by the validator, so it is internal to the deviating searcher--validator coalition. However, under the assumptions above, copied-code submissions cannot outbid the original submission or obtain higher tie-breaking priority in a tie.

\begin{theorem}[Ex-post validator--searcher coalition resistance of the FPA]
Under the assumptions above, the FPA is ex-post validator--searcher coalition-resistant under copied-code deviations.
\end{theorem}

\begin{proof}
Fix an opportunity \(o\), a real searcher \(i\), an original submission \(b\in\mathcal B_i(o)\), and a realized rival submission set \(\mathcal B_{-i}(o)\). As in the proof of ex-post copied-code Sybil resistance, any copied-code deviation \(\sigma_i=(\hat b_1,\ldots,\hat b_k)\) satisfies
\[
\beta_{i,\tau(o)}(z_{\hat b_\ell})
\le
\beta_{i,\tau(o)}(z_b)
\quad
\text{for every } \ell.
\]
Hence no copied submission can outbid the original submission; and if it ties the original submission, it cannot have higher tie-breaking priority.

If \(b\) loses, every copied submission also loses, so the validator payoff is determined by the same rival winner and the deviating searcher receives zero. If \(b\) wins, the searcher--validator combined payoff is \(V_o z_b\). Under a copied-code deviation, if a copied submission wins, its certified normalized net value is at most \(z_b\), so the combined payoff is at most \(V_o z_b\). If instead a rival submission wins, then its transfer offer is at most the original transfer offer of \(b\), which is at most \(V_o z_b\) by feasibility of \(\beta_{i,\tau(o)}\). Therefore, the copied-code deviation cannot increase the combined payoff of the deviating searcher and the validator.
\end{proof}

\begin{corollary}[Bayesian validator--searcher coalition resistance of the FPA]
Under the assumptions of the preceding theorem, the FPA is Bayesian validator--searcher coalition-resistant under copied-code deviations.
\end{corollary}

These results are specific to copied-code deviations under the same searcher-specific bidding function and the fixed-priority tie-breaking rule. The FPA remains non-incentive-compatible as a bidding mechanism: a searcher may prefer a shaded transfer offer to bidding its certified net value. However, once all submissions controlled by the same real searcher use the same feasible bidding function, with non-decreasing bids and non-decreasing retained surplus, copied or degraded submissions cannot outbid the original submission, improve tie-breaking priority, or increase retained surplus. Hence copied-code identity splitting does not improve either the searcher's payoff or the searcher--validator combined payoff.

\section{Relation to the Sybil-Proofness Results of Pan et al.}
\label{app:pan-comparison}

We briefly clarify the relation between our FPA security result and the characterization of Pan et al.~\cite{Pan2026SybilProof}, and discuss the applicability of their Bayesian construction to the value distributions considered in this paper.

\subsection{FPA and Ex-Post Sybil Resistance}
Pan et al. show that a payment-normalized, non-wasteful, symmetric, incentive-compatible, and ex-post Sybil-proof direct mechanism must be a second-price auction with symmetric tie-breaking. This does not contradict the ex-post copied-code Sybil resistance of the FPA established in this paper. The first-price auction does not satisfy their incentive-compatibility requirement as truthful bidding is generally not a dominant strategy in the FPA. Indeed, under the Sybil-proofness axiom used by Pan et al., the FPA itself is resistant to adding identities when the original identity continues to bid its true value \(v_i\): an additional bid below \(v_i\) cannot beat the original bid, whereas an additional bid at or above \(v_i\) gives non-positive utility if it wins. Thus, it is incentive compatibility, rather than Sybil-proofness, that excludes the FPA from their characterization.

\subsection{Bayesian Lottery and Truncated-Lognormal Values}
Pan et al. relax ex-post Sybil-proofness and construct a Bayesian Sybil-proof mechanism using an increasing ticket price. To avoid confusion with our searcher-specific distributions, let \(G\) and \(g\) denote the common value CDF and density in their construction, with upper endpoint \(\bar z\). With \(n\) submitted bids, their ticket price is
\[
\begin{split}
& a_n
=
\max\left\{
y^\star,\,
\mathbb E[Z_{n-1}],\,
G^{-1}\!\left(\frac{n-2}{n-1}\right)
\right\},
\\
& y^\star
=
\min\left\{
y:
(\bar z-x)g(x)\le 1-G(x)
\ \forall x\in[y,\bar z]
\right\},   
\end{split}
\]
where \(Z_{n-1}\) is the maximum of \(n-1\) independent draws from \(G\). If at least two bids are at least \(a_n\), the mechanism runs a uniform lottery among those bidders at price \(a_n\); otherwise it uses a second-price auction~\cite{Pan2026SybilProof}.

To relate this condition to our model, fix a searcher class and consider
\[
C\sim
\mathrm{TruncLogNormal}(\mu,\sigma^2;[\underline c,\overline c]),
\qquad
z^\star=1-C,
\]
so that \(\bar z=1-\underline c\). When
\[
\underline c<e^{\mu-\sigma^2},
\]
the log-normal cost density is increasing near its lower endpoint, and hence the induced value density \(g\) is decreasing near \(\bar z\). This is the economically natural case in our setting, since the minimum feasible execution cost is expected to lie below the modal (most likely) execution cost. For \(x\) sufficiently close to \(\bar z\),
\[
1-G(x)
=
\int_x^{\bar z}g(t)\,dt
<
(\bar z-x)g(x),
\]
which violates the upper-tail condition required for any \(y^\star<\bar z\). Consequently \(y^\star=\bar z\), and therefore \(a_n=\bar z\). Since the value distribution is continuous, \(\Pr(z^\star\ge a_n)=0\). The lottery branch then occurs with probability zero, so the increasing-ticket-price mechanism reduces almost surely to the standard second-price auction.

\section{Security Analysis of the Entry-Filtered SCA}
\label{app:sca-theta-security}
The entry-filtered SCA distributes rewards only among admitted submissions. A copied-code deviation may leave zero, one, or multiple copied submissions admitted after the entry filter is applied. We therefore characterize security through Bayesian interim payoff inequalities, with all payoff functionals evaluated on the resulting admitted profile.

Fix a searcher \(i\) with realized certified normalized value \(z_{i,o}^\star=z\), and let \(\mathcal B_{-i}\) denote the random rival submission profile induced by the participation and value primitives. Let \(b_i(z)\) denote \(i\)'s original submission. For \(q\ge2\), define the feasible copied-value set as \(\mathcal Z_q^c(z):=\{(z_1^c,\ldots,z_q^c):0\le z_\ell^c\le z,\ \ell=1,\ldots,q\}\). For \(\mathbf z^c=(z_1^c,\ldots,z_q^c)\in\mathcal Z_q^c(z)\), let \(C_{i,q}(\mathbf z^c)\) denote \(q\) copied-code submissions with certified normalized values \(z_1^c,\ldots,z_q^c\).

Let \(r_\theta(\mathbf z^c):=\sum_{\ell=1}^{q}\mathbf 1\{z_\ell^c\ge\theta\}\) denote the number of copied submissions that pass the entry filter. Here, \(q\ge2\) is the number of submitted Sybil identities, whereas \(r_\theta(\mathbf z^c)\in\{0,\ldots,q\}\) is the number of reward-relevant copies. In particular, \(r_\theta(\mathbf z^c)=0\) is payoff-equivalent to non-submission, while \(r_\theta(\mathbf z^c)=1\) is payoff-equivalent to replacing the original submission by one admitted, possibly degraded, submission.

Define the original-submission and worst-case \(q\)-copy interim searcher payoffs by
\[
\mathcal T_i^{S,\theta}(z;\kappa_S)
:=
\mathbb E\!\left[
U_i^{\mathrm{SCA}_\theta}
\bigl(\{b_i(z)\}\cup\mathcal B_{-i};\kappa_S\bigr)
\mid
D_{i,o}=1,\ z_{i,o}^\star=z
\right],
\]
and
\[
\mathcal S_{i,q}^{S,\theta}(z;\kappa_S)
:=
\sup_{\mathbf z^c\in\mathcal Z_q^c(z)}
\mathbb E\!\left[
U_i^{\mathrm{SCA}_\theta}
\bigl(C_{i,q}(\mathbf z^c)\cup\mathcal B_{-i};\kappa_S\bigr)
\mid
D_{i,o}=1,\ z_{i,o}^\star=z
\right].
\]

For validator--searcher coalitions, define
\[
\begin{aligned}
\mathcal T_i^{C,\theta}(z;\kappa_S,\kappa_V)
:=
\mathbb E\!\Big[
&
U_i^{\mathrm{SCA}_\theta}
\bigl(\{b_i(z)\}\cup\mathcal B_{-i};\kappa_S\bigr)
+
U_{\mathrm{val}}^{\mathrm{SCA}_\theta}
\bigl(\{b_i(z)\}\cup\mathcal B_{-i};\kappa_V\bigr)
\\
&\mid
D_{i,o}=1,\ z_{i,o}^\star=z
\Big],
\end{aligned}
\]
and
\[
\begin{aligned}
\mathcal S_{i,q}^{C,\theta}(z;\kappa_S,\kappa_V)
:=
\sup_{\mathbf z^c\in\mathcal Z_q^c(z)}
\mathbb E\!\Big[
&
U_i^{\mathrm{SCA}_\theta}
\bigl(C_{i,q}(\mathbf z^c)\cup\mathcal B_{-i};\kappa_S\bigr)
+
U_{\mathrm{val}}^{\mathrm{SCA}_\theta}
\bigl(C_{i,q}(\mathbf z^c)\cup\mathcal B_{-i};\kappa_V\bigr)
\\
&\mid
D_{i,o}=1,\ z_{i,o}^\star=z
\Big].
\end{aligned}
\]

If \(z<\theta\), the original submission and every feasible copied submission are rejected. Hence \(\mathcal T_i^{S,\theta}(z;\kappa_S)=\mathcal S_{i,q}^{S,\theta}(z;\kappa_S)=0\). Moreover, both the original and deviating coalition profiles induce the same admitted rival profile, so
\[
\mathcal T_i^{C,\theta}(z;\kappa_S,\kappa_V)
=
\mathcal S_{i,q}^{C,\theta}(z;\kappa_S,\kappa_V)
=
\mathbb E\!\left[
U_{\mathrm{val}}^{\mathrm{SCA}_\theta}
(\mathcal B_{-i};\kappa_V)
\mid
D_{i,o}=1,\ z_{i,o}^\star=z
\right].
\]
Thus, only \(z\ge\theta\) generates nontrivial security constraints.

\begin{theorem}[Bayesian security characterization]
\label{lem:bayesian-security-entry}
The entry-filtered SCA satisfies Bayesian copied-code Sybil resistance and Bayesian validator--searcher coalition resistance if and only if, for every searcher \(i\), every realized value \(z\) in its support, and every \(q=2,\ldots,H+1\),
\[
\mathcal T_i^{S,\theta}(z;\kappa_S)
\ge
\mathcal S_{i,q}^{S,\theta}(z;\kappa_S),
\qquad
\mathcal T_i^{C,\theta}(z;\kappa_S,\kappa_V)
\ge
\mathcal S_{i,q}^{C,\theta}(z;\kappa_S,\kappa_V).
\]
\end{theorem}

\begin{proof}
First suppose that the entry-filtered SCA satisfies Bayesian copied-code Sybil resistance and Bayesian validator--searcher coalition resistance. Then no copied-code deviation can increase either the deviating searcher's interim payoff or the searcher--validator coalition's interim payoff. Since \(\mathcal S_{i,q}^{S,\theta}\) and \(\mathcal S_{i,q}^{C,\theta}\) are defined as suprema over all \(\mathbf z^c\in\mathcal Z_q^c(z)\), the displayed inequalities must hold for every \(i\), every \(z\), and every \(q=2,\ldots,H+1\).

Conversely, suppose that the displayed inequalities hold for every \(i\), every \(z\), and every \(q=2,\ldots,H+1\). Consider any copied-code deviation with \(q\ge2\) submissions and copied-value profile \(\mathbf z^c\in\mathcal Z_q^c(z)\). Let \(r:=r_\theta(\mathbf z^c)\) be the number of admitted copies. Copies below \(\theta\) do not affect the admitted profile or any payoff.

If \(r\le H+1\), the same admitted profile can be represented using \(q'=\max\{2,r\}\le H+1\) submitted copies, adding rejected copies when \(r<2\). This representation includes the withdrawal-equivalent case \(r=0\) and the one-admitted-copy case \(r=1\). Its searcher and coalition payoffs are therefore bounded above by \(\mathcal S_{i,q'}^{S,\theta}(z;\kappa_S)\) and \(\mathcal S_{i,q'}^{C,\theta}(z;\kappa_S,\kappa_V)\), respectively. The assumed inequalities rule out profitability.

It remains to consider \(r>H+1\). The admitted copies alone then force the fallback branch for every rival profile. Retain only \(H+1\) admitted copies, including the copy with the highest certified value and, among ties, the highest tie-breaking priority. This preserves the deviator's best possible winning copy while weakly reducing the total admitted population \(t_\theta\). In the fallback branch, the validator payoff depends on the highest admitted certified value but not on \(t_\theta\), and is therefore unchanged. The deviating searcher's payoff is zero if it does not win; if it wins, its payment is proportional to \(1/(t_\theta+1)\), which weakly increases when copies are removed. Thus, the original deviation yields no larger searcher or coalition payoff than a deviation with \(H+1\) admitted copies.

Every copied-code deviation is therefore bounded above by a deviation covered by \(q=2,\ldots,H+1\). Since all such deviations satisfy the displayed inequalities, no copied-code deviation is profitable for either the searcher alone or the searcher--validator coalition.
\end{proof}

The feasible security-calibration region is
\[
\mathcal K_{\mathrm{sec}}^\theta
=
\left\{
(\kappa_S,\kappa_V)\in\mathbb R_+^2:
\begin{array}{l}
\mathcal T_i^{S,\theta}(z;\kappa_S)
\ge
\mathcal S_{i,q}^{S,\theta}(z;\kappa_S),\\
\mathcal T_i^{C,\theta}(z;\kappa_S,\kappa_V)
\ge
\mathcal S_{i,q}^{C,\theta}(z;\kappa_S,\kappa_V),\\
\forall i,\ 
\forall z\in\operatorname{supp}(F_{i,\tau(o)}),\
\forall q\in\{2,\ldots,H+1\}
\end{array}
\right\}.
\]
If \(\mathcal K_{\mathrm{sec}}^\theta\) is nonempty and the optimum is attained, a least-burn implementation selects a pair in \(\mathcal K_{\mathrm{sec}}^\theta\) that maximizes expected total paid surplus, equivalently minimizing expected burn subject to both Bayesian security requirements.

\section{Additional Numerical Analysis and Validator-Side Security Calibration}
\label{app:revenue-decomposition}

\subsection{Sensitivity to Elite Execution Advantage}

Figure~\ref{fig:centralization-sca-fpa} compares the FPA and the entry-filtered SCA as the execution-quality advantage of a single elite searcher increases. We consider a three-class truncated-lognormal environment with one elite, nine mid-tier, and ten tail searchers. The mid-tier class has participation probability \(0.55\), median cost \(\tilde c_{\mathrm{mid}}=0.30\), and support \([0.12,0.55]\), while the tail class has participation probability \(0.35\), median cost \(0.55\), and support \([0.25,0.85]\). The elite searcher has participation probability \(0.95\) and support \([0.02,0.70]\).

We vary the elite cost advantage
\[
A:=\frac{\tilde c_{\mathrm{mid}}}{\tilde c_{\mathrm{elite}}},
\qquad
\tilde c_{\mathrm{elite}}=\frac{0.30}{A},
\]
so that larger \(A\) corresponds to a stronger elite execution-quality advantage. For each \(A\), the entry threshold is chosen so that
\[
\mathbb E[N_o^\theta]=6,
\]
the cap is set to \(H=9\), and \(\kappa_S^\star\) is calibrated against the full Bayesian copied-code deviation set over all realized values \(z\), copy sizes \(q=2,\ldots,H+1\), and feasible copied-value profiles \(\mathbf z^c\in\mathcal Z_q^c(z)\).

\begin{figure}[t]
\centering
\includegraphics[width=0.8\textwidth]{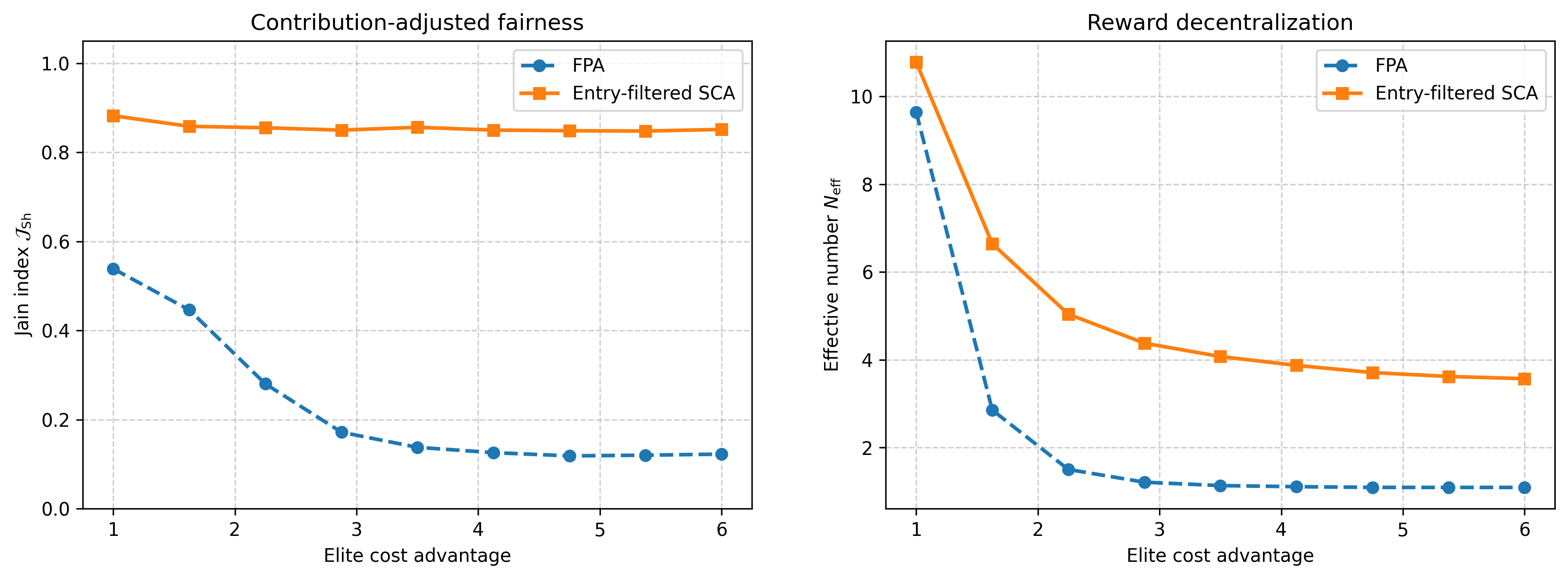}
\caption{Effect of increasing elite execution advantage on fairness and reward decentralization.}
\label{fig:centralization-sca-fpa}
\end{figure}

\subsection{Revenue Decomposition and Validator-Side Security Calibration}

This appendix calibrates the validator-side security parameter of the entry-filtered SCA and decomposes gross opportunity value into execution cost, searcher payments, validator payments, and burn. The calibration imposes Bayesian coalition resistance over the full copied-value deviation set defined in Appendix~\ref{app:sca-theta-security}, including withdrawal-equivalent, one-admitted-copy, and identity-expansion deviations.

We consider the centralized, low-admission environment summarized in Table~\ref{tab:validator-calibration-environment}. It consists of one elite searcher, seven mid-tier searchers, and ten tail searchers, with class-specific participation probabilities and truncated-lognormal execution-cost distributions. Throughout this calibration, \(\sigma=0.35\) and \(z_{i,o}^\star=1-C_i\).

\begin{table}[b]
\centering
\caption{Centralized low-admission environment used for validator-side security calibration.}
\label{tab:validator-calibration-environment}
\small
\begin{tabular}{l c c c}
\hline
\textbf{Class}
& \(\boldsymbol{\#}\)
& \(\boldsymbol{\alpha}\)
& \(\boldsymbol{\tilde c,[\underline c,\overline c]}\) \\
\hline
Elite & 1  & 0.98 & \(0.08,\,[0.02,0.25]\) \\
Mid   & 7  & 0.45 & \(0.30,\,[0.15,0.55]\) \\
Tail  & 10 & 0.15 & \(0.55,\,[0.35,0.85]\) \\
\hline
\end{tabular}
\end{table}

Under these parameters, the raw expected active count is \(\mathbb E[N_o]=5.63\). We choose the entry threshold so that \(\mathbb E[N_o^\theta]=4\), which yields \(\theta=0.5730129\), and set the cap to \(H=7\).

We first fix the searcher-side parameter at the smallest value satisfying all Bayesian copied-code Sybil-resistance constraints, \(\kappa_S^\star=0.2137025\). For a searcher \(i\) with realized value \(z\), a copy size \(q\), and a feasible copied-value profile \(\mathbf z^c\in\mathcal Z_q^c(z)\), define the Bayesian coalition gain by
\[
\begin{aligned}
\Gamma_{i,q}^{C,\theta}
(z,\mathbf z^c;\kappa_V)
:=
\mathbb E\!\Big[
&
U_i^{\mathrm{SCA}_\theta}
\bigl(C_{i,q}(\mathbf z^c)\cup\mathcal B_{-i};
      \kappa_S^\star\bigr)
+
U_{\mathrm{val}}^{\mathrm{SCA}_\theta}
\bigl(C_{i,q}(\mathbf z^c)\cup\mathcal B_{-i};
      \kappa_V\bigr)
\\
&\mid
D_{i,o}=1,\ z_{i,o}^\star=z
\Big]
-
\mathcal T_i^{C,\theta}
(z;\kappa_S^\star,\kappa_V).
\end{aligned}
\]
The copied-value profile is selected before the rival profile \(\mathcal B_{-i}\) is realized. Accordingly, Bayesian validator--searcher coalition resistance requires
\[
\sup_{\substack{
i,\;
z\in\operatorname{supp}(F_{i,\tau(o)})\cap[\theta,\infty),\\
q=2,\ldots,H+1,\;
\mathbf z^c\in\mathcal Z_q^c(z)
}}
\Gamma_{i,q}^{C,\theta}
(z,\mathbf z^c;\kappa_V)
\le0.
\]
Because copied identities are interchangeable, the numerical maximization may, without loss of generality, restrict attention to profiles satisfying \(z_1^c\ge\cdots\ge z_q^c\).

Recall that \(r_\theta(\mathbf z^c)\) denotes the number of admitted copies. The full coalition constraint produces opposing restrictions on \(\kappa_V\).

First, deviations with \(r_\theta(\mathbf z^c)\ge2\) expand the admitted population relative to the original single submission. Additional identities strengthen the payment discounts and generally increase burn, but may also increase the coalition's gross mechanism allocation. Consequently, \(\kappa_V\) must be sufficiently large to prevent the gross coalition gain from dominating the additional discount.

In this environment, the binding identity-expansion deviation is made by a Tail searcher at \(z=\theta\), with \(q=7\), \(r_\theta(\mathbf z^c)=7\), and \(\mathbf z^c=(\theta,\ldots,\theta)\). Let \(M\) denote the number of admitted rival submissions. Under this deviation, \(t_\theta=M+7\). Hence, the Shapley branch applies only when \(M=0\), while the fallback branch applies whenever \(M\ge1\). Both cases, together with the realized rival values, are included in the Bayesian conditional expectation. This deviation imposes the lower bound \(\kappa_V\ge\underline\kappa_V=0.0618379\). At the boundary, the original-submission and deviating normalized Bayesian coalition payoffs are both approximately \(0.550652\).

Profiles with \(r_\theta(\mathbf z^c)=1\), including deviations with one admitted degraded copy and any number of rejected copies, are included in the supremum but do not bind in this calibration.

Second, the feasible set permits \(r_\theta(\mathbf z^c)=0\), in which every copied submission falls below \(\theta\). Although formally represented by \(q\ge2\) copied submissions, this deviation is payoff-equivalent to non-submission by the searcher. Removing the searcher may reduce the coalition's gross mechanism allocation, but it also reduces the admitted population, weakens the validator discount, and lowers burn. If \(\kappa_V\) is too large, the reduction in discounting dominates and makes the withdrawal-equivalent deviation profitable. The binding upper-bound deviation is associated with a Tail searcher at \(z=\theta\) and imposes \(\kappa_V\le\overline\kappa_V=0.0665793\). At this boundary, the original-submission and deviating normalized Bayesian coalition payoffs are both approximately \(0.540782\).

The numerically identified Bayesian coalition-security region, conditional on \(\kappa_S=\kappa_S^\star\), is
\[
\mathcal K_V^{C,\theta}(\kappa_S^\star)
=
[0.0618379,\;0.0665793].
\]
Identity expansion imposes the lower bound on \(\kappa_V\), whereas the withdrawal-equivalent deviation imposes the upper bound. Since the honest validator payment is nonincreasing in \(\kappa_V\), the least-burn secure calibration conditional on \(\kappa_S^\star\) is the lower endpoint,
\[
{\kappa_V^\star=0.0618379.}
\]
The conditional expectations are evaluated over the participation and truncated-lognormal value primitives, and the supremum over copied-value profiles is solved numerically. The reported endpoints are stable to six decimal places under the stated grid refinement.

To report the resulting value decomposition, define \(\overline z_t:=\mathbb E[z_{(1)}^\theta\mid t_\theta=t]\) and
\[
\overline\Phi_{\mathrm{val},t}
:=
\mathbb E\!\left[
\frac{\Phi_{\mathrm{val}}^\theta(z_o^\theta)}{V_o}
\Bigm|
t_\theta=t
\right].
\]
For \(t\le H\), the conditional expected normalized payments are
\[
S_t
=
e^{-\kappa_S^\star(t-1)}
\left(
\overline z_t-\overline\Phi_{\mathrm{val},t}
\right),
\qquad
V_t
=
e^{-\kappa_V^\star(t-1)}
\overline\Phi_{\mathrm{val},t}.
\]
For \(t>H\), the fallback payments are
\[
S_t
=
e^{-\kappa_S^\star(H-1)}
\frac{\overline z_t}{H(t+1)},
\qquad
V_t
=
e^{-\kappa_V^\star(H-1)}
\frac{H}{H+1}\overline z_t.
\]
In both branches, \(B_t=\overline z_t-S_t-V_t\) and \(C_t=1-\overline z_t\), where \(B_t\) is burn and \(C_t\) is execution cost.

Table~\ref{tab:sca-theta-gross-value-decomposition} reports the conditional decomposition for \(t_\theta=1,\ldots,H+1\), where \(t_\theta=H+1\) is the first fallback state.

\begin{table}[t]
\centering
\caption{Conditional gross-value decomposition under the entry-filtered SCA with \(\kappa_S^\star=0.2137025\) and \(\kappa_V^\star=0.0618379\). All entries are normalized by \(V_o\) and conditional on the admitted count \(t_\theta\).}
\label{tab:sca-theta-gross-value-decomposition}
\resizebox{\textwidth}{!}{
\begin{tabular}{c c c c c c}
\hline
\(\boldsymbol{t_\theta}\)
& \textbf{Branch}
& \textbf{Searcher payment}
& \textbf{Validator payment}
& \textbf{Burn/loss}
& \textbf{Execution cost}
\\
\hline
1 & Shapley  & 0.4483 & 0.4483 & 0.0000 & 0.1035\\
2 & Shapley  & 0.2717 & 0.5376 & 0.0990 & 0.0918\\
3 & Shapley  & 0.1798 & 0.5619 & 0.1699 & 0.0884\\
4 & Shapley  & 0.1251 & 0.5612 & 0.2268 & 0.0869\\
5 & Shapley  & 0.0898 & 0.5488 & 0.2753 & 0.0861\\
6 & Shapley  & 0.0659 & 0.5303 & 0.3181 & 0.0856\\
7 & Shapley  & 0.0492 & 0.5087 & 0.3567 & 0.0853\\
8 & Fallback & 0.0040 & 0.5524 & 0.3585 & 0.0851\\
\hline
\end{tabular}}
\end{table}

Before rounding, every row satisfies \(C_t+S_t+V_t+B_t=1\). Within the Shapley branch, the searcher discount decreases rapidly with \(t_\theta\), causing searcher payments to fall and burn to rise. The validator payment is nonmonotone because a larger admitted population may increase the validator's gross allocation while simultaneously strengthening its discount. At \(t_\theta=H+1\), the fallback branch sharply reduces the searcher payment and assigns the validator a fixed discounted fraction of the highest certified net value.

\end{document}